\documentclass[11pt]{article}

\nonstopmode
\usepackage[margin=1in]{geometry}
\usepackage[english]{babel}
\usepackage{xr}

\usepackage[round]{natbib}

\usepackage{setspace}
\usepackage{tikz}
\usetikzlibrary{arrows,automata,positioning}
\usepackage{pgfplots}
\pgfplotsset{compat=newest}

\usepackage{mathrsfs,hyperref}
\usepackage[]{amsmath}
\usepackage{amsthm}
\usepackage{fix-cm}
\usepackage[]{amssymb}
\usepackage[]{latexsym}
\usepackage[right]{eurosym}
\usepackage[T1]{fontenc}
\usepackage[]{graphicx}
\usepackage[]{epsfig}
\usepackage{fancyhdr}
\usepackage{multirow}
\usepackage{subcaption}

\allowdisplaybreaks

\makeatletter

\newcommand{\bbr}{\mathbb{R}}

\newcommand{\fn}{\footnote}

\newcounter{modcount}
\newcommand{\modulo}[2]{%
\setcounter{modcount}{#1}\relax
\ifnum\value{modcount}<#2\relax
\else\relax
\addtocounter{modcount}{-#2}\relax
\modulo{\value{modcount}}{#2}\relax
\fi}
\newcommand{\tablepictures}[4][c]{\begin{tabular}[#1]{@{}c@{}}#2\vspace{0.5cm}\\(\alph{#4}) #3\end{tabular}}
\newcounter{gridsearch}
\newcommand{\tabpic}[2]{
    \stepcounter{gridsearch}
    \modulo{\thegridsearch}{2}
    \ifnum\value{modcount}=0
        \tablepictures[t]{#1}{#2}{gridsearch}\\[2.0cm]
    \else
        \tablepictures[t]{#1}{#2}{gridsearch}&~&
    \fi
}

\makeatother
\newtheorem{lemma}{Lemma}[section]
\newtheorem{proposition}[lemma]{Proposition}
\newtheorem{theorem}[lemma]{Theorem}
\newtheorem{corollary}[lemma]{Corollary}

\newtheorem{example1}[lemma]{Example}
\newtheorem{rem1}[lemma]{Remark}
\newtheorem{assumption}[lemma]{Assumption}
\newtheorem{alg1}[lemma]{Algorithm}
\newtheorem{me1}[lemma]{Mechanism}
\newtheorem{study1}[lemma]{Case Study}

\newenvironment{remark}{\begin{rem1}\rm}{\end{rem1}}
\newenvironment{example}{\begin{example1}\rm}{\end{example1}}

\newenvironment{study}{\begin{study1}\rm}{\end{study1}}

\usepackage{color}

\newcommand{\T}{\top}

\newcommand\ind[1]{\mathbb{I}_{\{#1\}}}

\begin{document}

\title{Contingent Convertible Obligations and Financial Stability}
\author{Zachary Feinstein \thanks{Stevens Institute of Technology, School of Business, Hoboken, NJ 07030, USA, {\tt zfeinste@stevens.edu}}
\and
T.R.\ Hurd \thanks{McMaster University, Department of Mathematics \& Statistics, Hamilton, ON L8S 4L8, Canada}}
\date{\today}
\maketitle
\abstract{
This paper investigates whether a financial system can be made more stable if financial institutions share risk by exchanging contingent convertible (CoCo) debt obligations.  The question is framed in a financial network model of debt and equity interlinkages with the addition of a variant of the CoCo that converts continuously when a bank's equity-debt ratio drops to a trigger level.  The main theoretical result is a complete characterization of the clearing problem for the interbank debt and equity at the maturity of the obligations.  We then introduce stylized networks to study when introducing contingent convertible bonds improves financial stability, as well as specific networks for which contingent convertible bonds do not  provide uniformly improved system performance.  To return to the main question, we examine the EU financial network at the time of the 2011 EBA stress test to do comparative statics to study the implications of CoCo debt on financial stability. It is found that by replacing all unsecured interbank debt by standardized CoCo interbank debt securities,  systemic risk in the EU will decrease and bank shareholder value will increase. 

\bigskip

}

\section{Introduction}\label{sec:intro}
Systemic risk is the risk of financial contagion -- the spread of losses from one bank or institution to other firms through interconnections in the financial system.  Financial contagion was exhibited most prominently in the global financial crisis of 2007-2009 and led to deep losses in the real economy. In the wake of that crisis, contingent convertible bonds (CoCos) were introduced into European markets. These are contingent debt claims that convert debt to equity at a pre-defined debt-equity ratio, and are designed to improve the solvency of the issuing company.  Specifically, CoCo bonds reduce the need for stressed firms to raise additional capital; these instruments additionally act as a mechanism to discipline banks ex-ante insofar as the conversion of debt to equity dilutes the shares of the original equity holders. As such, CoCo bonds have become prominent in discussions on financial regulation:  We refer to~\cite{glasserman2012coco,AKB2013coco,SMS2018coco} for detailed discussion of specific kinds of convertible bonds and methodologies for pricing these obligations in a single institution setting. 

The motivation of the present paper is to understand the implications of contingent convertible bonds on financial stability. Until recently, most systemic risk models were variations of the network clearing model first proposed by~\cite{EN01}, in which  firms have only  ``vanilla'' external and interbank assets and  liabilities.  Amongst the notable ``fixed network models'', ~\cite{RV13} accounts for bankruptcy costs, and \cite{suzuki2002,gourieroux2012} model equity cross-holdings between financial institutions. The related work of~\cite{AW_15}  also considers the joint impacts of bankruptcy costs and equity cross-holdings, assuming fractional recovery of the value of equity in case it is sold to pay off liabilities. 

The fixed network assumption no longer holds when derivatives such as contingent convertible bonds (CoCos) are included.  \cite{gupta2019coco,li2020default,li2020impact,calice2020contingent,balter2020contingent} consider contingent convertible debts within interbank payment networks under differing modeling setups; \cite{gupta2019coco}, for instance, presents a case study of the use of contingent convertible debts for reducing systemic risk when all such obligations are due outside the financial system. For networks with contingent payments based on generalized credit default swap obligations,~\cite{SSB16b,SSB17} show that the clearing payments may not be well-defined and, in fact, may not exist outside of specific network topologies.
\cite{minca2018reinsurance} studies reinsurance networks and is able to determine the existence and uniqueness of the realized liabilities and clearing payments between firms in the system.
\cite{banerjee2017insurance} presents a model in the vein of~\cite{BBF18} that is a time dynamic extension of~\cite{EN01} for interbank liabilities with  insurance obligations contingent on the state of the financial system.  

 In contrast to these cited papers,  we are able to characterize network clearing with contingent convertible bonds as an Eisenberg-Noe type fixed point condition on a single vector of firm assets, rather than joint payment and equity vectors as studied in~\cite{suzuki2002,gourieroux2012,AW_15}. By adopting a CoCo variant introduced  in~\cite{glasserman2012coco} with  \emph{fractional} conversion instead of the more usual all-or-nothing conversion, we avoid the possible non-existence of a fixed point due to cycles in which either the bond is converted or not,  as detailed in~\cite{KV16,banerjee2017insurance,balter2020contingent}.  With the innovation of fractional CoCos utilized in an interbank network, our first main contribution proves general conditions for the existence of maximal and minimal clearing solutions for the contingent network model with bankruptcy costs.   
 
The remaining contributions of this paper apply the contingent network framework to understand the conditions under which CoCo bonds improve financial stability. We provide a general discussion of when CoCo liabilities do and do not reduce the number of defaulting banks.  Extending analogous results of~\cite{AOT15} for vanilla interbank networks, we find that CoCo bonds -- either interbank or external -- improve financial stability in a symmetric and completely connected (hereafter called ``strongly symmetric'') system with idiosyncratic shocks. 
Finally, we undertake a sequence of case studies on financial stability within the EU financial system as it appeared in the 2011. Based on the 2011 EBA stress testing database, we study the effects of different counterfactual ways contingent convertible debts might be implemented.  This study leads to the observation that replacing unsecured interbank debt in the EU financial system by CoCo securities will usually improve systemic risk. These considerations allow us to test the heuristics imposed in the Basel III Accords which discourage the holding of interbank CoCo bonds in order to reduce the risk of financial contagion (see, e.g.,~\cite{AKB2013coco}).

The organization of this paper is as follows.  First, in Section~\ref{sec:1bank}, we present  contingent convertible bonds for a single bank, with an emphasis on the CoCo variant with fractional conversion. We show for the first time that there is a piecewise {\em non-linear} relationship between the value of a CoCo with fractional conversion and the external asset value.  We extend this setting to a network model of contingent convertible bonds in Section~\ref{sec:nbank} and present the main result of this work -- the existence of maximal and minimal clearing solutions for the interbank network with contingent convertible bonds.  We expound upon this result in Section~\ref{sec:discussion} with a general discussion of when contingent convertible bonds improve financial stability.  We prove that increased use of CoCo bonds systematically decreases the probability of defaults in a strongly symmetric financial system. However, we also provide an example of a small network showing that introducing such financial instruments may increase the number of defaults in the system.  
In Section~\ref{sec:eba}, we model the EU financial system based on the 2011 EBA stress testing data, and study system stability when different schemes for contingent convertible bonds are introduced. In Section~\ref{sec:conclusion}, our main conclusion is that a regulatory intervention that replaces all unsecured interbank debt in the EU financial system by a standardized CoCo security will create a win-win situation in which systemic risk is decreased while bank shareholder value increases. 

\noindent {\bf Notation:\ } The following  notation will be adopted throughout  this paper.  Let $x,y \in \bbr^n$ for some positive integer $n$, then
\begin{align*}
x \wedge y &= \left(\min(x_1,y_1),\min(x_2,y_2),\ldots,\min(x_n,y_n)\right)^\T,\\
x \vee y &= \left(\max(x_1,y_1),\max(x_2,y_2),\ldots,\max(x_n,y_n)\right)^\T,
\end{align*}
$x^+ = x \vee 0$, and $x^- = -(x \wedge 0)$.  Further, to ease notation, we will denote $[x,y] := [x_1,y_1] \times [x_2,y_2] \times \ldots \times [x_n,y_n] \subseteq \bbr^n$ to be the $n$-dimensional compact interval for $y - x \in \bbr^n_+$.  Similarly, we will consider $x \leq y$ if and only if $y - x \in \bbr^n_+$.

\section{A Bank with Contingent Convertible Obligations}\label{sec:1bank}

As a lead into our study of a network of ``banks'' (meaning the collection of financial institutions), in this section, we consider a single bank whose total assets are given by $a\ge 0$ and whose total debt obligations are structured as a combination of pure vanilla debt and a single class of contingent convertible (CoCo) debt with fractional conversion, as first introduced by~\cite{glasserman2012coco}.  This type of CoCo is parametrized by a trigger level $\tau > 0$ and conversion factor $q \in [0,1]$. In a discrete time setting, it converts a prescribed fraction of debt to equity at a set of conversion dates  $0 \leq T_1 < T_2 < \dots < T_K$, until the final maturity of the claims at $T_K$, if on these dates the equity of the issuer drops below $\tau$ times its debt. For each \$1 of CoCo debt converted to equity, the total firm equity increases by \$1 and the CoCo investor receives {\em new} equity shares with value  $q$. Finally, bankruptcy of the bank is defined to occur when the firm equity $E$ becomes negative: in this event, bankruptcy charges amounting to a fraction $1-\alpha$ of the firm's assets are paid, and CoCo holders and original shareholders receive nothing.

In this section, we summarize formulas based on \cite{glasserman2012coco} for the value of vanilla debt, CoCo debt and equity shares at the maturity  date $T=T_K$, as explicit functions of the independent variable, the firm's assets $a\in \bbr_+$. The remaining parameters are $\bar p^0\ge 0$, $\bar p^c\ge 0$, the face values of vanilla and CoCo debt  respectively. In the following we let $\lambda\in[0,1]$ denote the fraction of the CoCo that is converted to equity. We warn the reader that our setting differs from \cite{glasserman2012coco} in several respects: (i) our bankruptcy condition is the limited liability assumption typically made in the systemic risk literature; (ii) we value claims only at the maturity date $T_K$; (iii) the trigger parameter is defined differently from their parameter $\alpha$, with the relation $\tau=\alpha/(1-\alpha)$.

It turns out that the valuation functions are continuous on four subintervals, $[0,a_1), [a_1,a_2], (a_2,a_3), [a_3,\infty)$, with the breakpoints  $a_1=\bar p^0, a_2=(1+\tau) \bar p^0, a_3= (1+\tau) (\bar p^0+\bar p^c)$. These intervals correspond to four possible scenarios for the bank: (i)  the bank has defaulted and $\lambda=1$; (ii) the bank is solvent, but full conversion has occurred so $\lambda=1$; (iii) partial conversion has occurred with $\lambda\in (0,1)$; (iv) no conversion has occurred so $\lambda=0$ and $E(a)>\tau(\bar p^0+\bar p^c)$. On $[0,a_1)$, the value of equity and debt are given by $E(a)=\alpha a-\bar p^0$ and $D(a)=\alpha a$. On $[a_1,a_2]$, the value of equity and debt are given by $E(a)=a-\bar p^0$ and $D(a)=\bar p^0$. 
On $(a_2,a_3)$, the value of equity and debt are given by $E(a)=a-\bar p(\lambda(a))$ and $D(a)=\bar p(\lambda(a))$, where $\bar p(\lambda)=\bar p^0+(1-\lambda)\bar p^c$ and $\lambda(a)=\frac{\bar p^0+\bar p^c-a/(1+\tau)}{\bar p^c}$. On $[a_3,\infty)$, the value of equity and debt are given by $E(a)=a-\bar p^0-\bar p^c$ and $D(a)=\bar p^0+\bar p^c$. These formulas are summarized by 
\begin{eqnarray}
\label{E} E(a) & = & \ind{a \geq a_1}a+ \ind{a < a_1}\alpha a- \bar p(\lambda(a))\\
\label{D} D(a) & = & \ind{a \geq a_1}\bar p(\lambda(a))+\ind{a < a_1}\alpha a\\
\label{bar-p} \bar p(\lambda)&=&\bar p^0+(1-\lambda)\bar p^c\\
\label{lambda} \lambda(a)&=&1\wedge\left(\frac{\bar p^0+\bar p^c-a/(1+\tau)}{\bar p^c}\right)^+
\end{eqnarray}
and satisfy the identities $a=\bar p(\lambda(a))+E(a)=D(a)+E(a)^+$. Note that the total value of debt $D(a)$ is divided into vanilla debt $\ind{a \geq a_1}\bar p^0+\ind{a < a_1}\alpha a$ and CoCo debt $\ind{a \geq a_1}(1-\lambda(a))\bar p^c$.

The above formulas do not determine how the firm equity is split between original shareholders and the CoCo holders. An adaptation of the proof of Theorem 3.1 of \cite{glasserman2012coco}  shows that the equity fractions held by  the CoCo holders and original shareholders respectively are $c, 1-c$ where $c=c(a)$ satisfies the following differential equation for $a\in(a_2,a_3)$:
\begin{equation}
\label{c_ODE}
\frac{dc}{ da} = \frac{q(1 - c)}{\tau a}\ .  
\end{equation}
This important result can be derived directly as follows.  Suppose $a, a+\Delta a$ are two points in $(a_2,a_3)$ with corresponding values $\lambda,  \lambda+\Delta \lambda$ with $ \Delta \lambda = -\frac{\Delta a}{(1+\tau) \bar p^c}$. Then the CoCo assumptions imply that the corresponding fractions $c, c+\Delta c$ held by the CoCo owners must satisfy an equation for the total equity held by CoCo owners with additional conversion $\bar p^c \Delta \lambda$ and firm assets $a+\Delta a$:
\[  ( c+\Delta c)(E(a)+\Delta a-\bar p^c\Delta \lambda)=  c(E(a)+\Delta a) -q \bar p^c\Delta \lambda -c(1-q)\bar p^c\Delta \lambda\]
The differential equation arises in the limit $\Delta a\to 0$. Solving \eqref{c_ODE} with $c(a_3)=0$ and noting that $c'(a)=0$ for $a\notin[ a_2,a_3]$ leads to the desired formula for all $a\in [0,\infty)$:
\begin{equation}
\label{c_solution}
c(a) := 1-\left(\frac{(a\vee a_2)\wedge a_3}{a_3}\right)^{\frac{q}{\tau}}
\end{equation}

The form \eqref{c_solution} of $c(a)$ also implies that CoCos with $q>1$ are not financially desirable: they encourage speculation because the CoCo value is not monotonically increasing with $a$ near the trigger level. This means CoCo holders would sometimes reap benefits as the firm wealth drops. The following result rules out such non-monotonicity for all stakeholders of the firm under the condition $q\le 1$, which is a natural \emph{non-speculative condition} in the sense of \cite{banerjee2017insurance}.

\begin{theorem}\label{nondecreasing} The three types of liability securities on the bank are valued by the following functions of the total asset $a\ge 0$. 
\begin{enumerate}
  \item Vanilla debt has total value $\ind{a \geq a_1}\bar p^0+\ind{a < a_1}\alpha a$;
  \item CoCo liabilities have total value $(1-\lambda(a))\bar p^c + c(a) E(a)^+$;
  \item Original shareholders have total value $(1-c(a)) E(a)^+$.
\end{enumerate}
Under the \emph{non-speculative condition} $q\le 1$, all three values are non-decreasing in $a$ for $a\ge 0$. 
\end{theorem}

\begin{proof} The stated valuation formulas are implicit in equations \eqref{E}, \eqref{c_solution}. We check the monotonicity condition on CoCo liabilities when $a\in(a_2,a_3)$, and leave verification of the remaining cases to the reader. For $a\in(a_2,a_3)$ we note that the bank's leverage ratio is $E/A=\frac{\tau}{1+\tau}$ and thus $(1-\lambda(a))\bar p^c=\frac{1}{1+\tau}a-\bar p^0$ and $E(a)=\frac{\tau}{1+\tau}a$. Then we find that 
\[\frac{d}{da}(1-\lambda(a))\bar p^c + c(a) E(a)^+ = \frac{1}{1+\tau}+\frac{\tau}{1+\tau}-\frac{q+\tau}{1+\tau}(a/a_3)^{q/\tau} >0\]
\end{proof}
\begin{example}\label{ex:setting-coco} Consider a bank with a stylized bank balance sheet consisting of vanilla liabilities $\bar p^0 = 10$ and CoCo liabilities $\bar p^c = 4$ structured with trigger level $\tau = 0.1$, meaning the CoCo bonds convert from debt to equity when the debt-equity ratio exceeds $1/\tau = 10$.  Suppose also the recovery rate at default is $\alpha=0.5$ while the conversion factor may take one of three values $q=0, 0.5, 1$.    
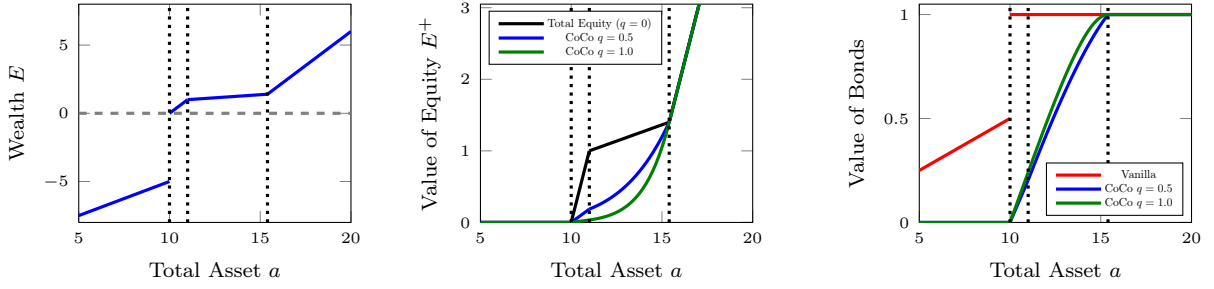
\begin{figure}[t]
\centering
\begin{subfigure}[t]{0.315\textwidth}
\centering
\begin{tikzpicture}
\begin{axis}[
    width=\textwidth,
    xlabel={Total Asset $a$},
    ylabel={Wealth $E$},
    xmin=5, xmax=20,
    ymin=-8, ymax=8,
    label style={font=\scriptsize},
    ticklabel style={font=\tiny},
    ]
\addplot [domain=5:10,solid,color=blue,very thick]{x/2-10};
\addplot [domain=10:11,solid,color=blue,very thick]{x-10};
\addplot [domain=11:15.4,solid,color=blue,very thick]{x/11};
\addplot [domain=15.4:20,solid,color=blue,very thick]{x-14};
\addplot [color=black,dotted,very thick] coordinates {
    (10,-8)
    (10,8)};
\addplot [color=black,dotted,very thick] coordinates {
    (11,-8)
    (11,8)};
\addplot [color=black,dotted,very thick] coordinates {
    (15.4,-8)
    (15.4,8)};
\addplot [color=gray,dashed,very thick] coordinates {
    (5,0)
    (20,0)};
\end{axis}
\end{tikzpicture}
\caption{The bank's equity $E$ as a function of the value of the total assets $a$.}
\label{fig:1bank-wealth}
\end{subfigure}
~
\begin{subfigure}[t]{0.315\textwidth}
\begin{tikzpicture}[declare function={ lam(\x) = (1.4-\x/11)/0.4; c(\l) = (1-(1-2/7*\l)^5); c1(\l) = (1-(1-2/7*\l)^10);},]
\begin{axis}[
    width=\textwidth,
    xlabel={Total Asset $a$},
    ylabel={Value of Equity $E^+$},
    xmin=5, xmax=20,
    ymin=0, ymax=3.05,
    label style={font=\scriptsize},
    ticklabel style={font=\tiny},
    legend pos=north west,
    legend style={nodes={scale=0.4, transform shape}},
    ]
\addplot [domain=5:10,solid,color=black,very thick]{0};
\addplot [domain=10:11,solid,color=black,very thick,forget plot]{x-10};
\addplot [domain=11:15.4,solid,color=black,very thick,forget plot]{x/11};
\addplot [domain=15.4:20,solid,color=black,very thick,forget plot]{x-14};
\addlegendentry{Total Equity ($q = 0$)}
\addplot [domain=5:10,solid,color=blue,very thick]{0};
\addplot [domain=10:11,solid,color=blue,very thick,forget plot]{(1-c(1))*(x-10)};
\addplot [domain=11:15.4,solid,color=blue,very thick,forget plot]{(1-c(lam(x)))*x/11};
\addplot [domain=15.4:20,solid,color=blue,very thick,forget plot]{x-14};
\addlegendentry{CoCo $q = 0.5$}
\addplot [domain=5:10,solid,color=black!50!green,very thick]{0};
\addplot [domain=10:11,solid,color=black!50!green,very thick,forget plot]{(1-c1(1))*(x-10)};
\addplot [domain=11:15.4,solid,color=black!50!green,very thick,forget plot]{(1-c1(lam(x)))*x/11};
\addplot [domain=15.4:20,solid,color=black!50!green,very thick,forget plot]{x-14};
\addlegendentry{CoCo $q = 1.0$}
\addplot [color=black,dotted,very thick] coordinates {
    (10,0)
    (10,3.05)};
\addplot [color=black,dotted,very thick] coordinates {
    (11,0)
    (11,3.05)};
\addplot [color=black,dotted,very thick] coordinates {
    (15.4,0)
    (15.4,3.05)};
\end{axis}
\end{tikzpicture}
\caption{The value of the bank's equity $E$ owned by the original shareholders as a function of the value of the total assets $a$.}
\label{fig:1bank-equity}
\end{subfigure}
~
\begin{subfigure}[t]{0.315\textwidth}
\centering
\begin{tikzpicture}[declare function={ lam(\x) = (1.4-\x/11)/0.4; c(\l) = (1-(1-2/7*\l)^5); c1(\l) = (1-(1-2/7*\l)^10);},]
\begin{axis}[
    width=\textwidth,
    xlabel={Total Asset $a$},
    ylabel={Value of Bonds},
    xmin=5, xmax=20,
    ymin=0, ymax=1.05,
    label style={font=\scriptsize},
    ticklabel style={font=\tiny},
    legend pos=south east,
    legend style={nodes={scale=0.4, transform shape}},
    ]
\addplot [domain=5:10,solid,color=red,very thick]{x/2/10};
\addplot [domain=10:20,solid,color=red,very thick,forget plot]{10/10};
\addlegendentry{Vanilla}
\addplot [domain=5:10,solid,color=blue,very thick]{0};
\addplot [domain=10:11,solid,color=blue,very thick,forget plot]{c(1)*(x-10)/4};
\addplot [domain=11:15.4,solid,color=blue,very thick,forget plot]{((1-lam(x))*4 + c(lam(x))*x/11)/4};
\addplot [domain=15.4:20,solid,color=blue,very thick,forget plot]{4/4};
\addlegendentry{CoCo $q = 0.5$}
\addplot [domain=5:10,solid,color=black!50!green,very thick]{0};
\addplot [domain=10:11,solid,color=black!50!green,very thick,forget plot]{c1(1)*(x-10)/4};
\addplot [domain=11:15.4,solid,color=black!50!green,very thick,forget plot]{((1-lam(x))*4 + c1(lam(x))*x/11)/4};
\addplot [domain=15.4:20,solid,color=black!50!green,very thick,forget plot]{4/4};
\addlegendentry{CoCo $q = 1.0$}
\addplot [color=black,dotted,very thick] coordinates {
    (10,0)
    (10,1.05)};
\addplot [color=black,dotted,very thick] coordinates {
    (11,0)
    (11,1.05)};
\addplot [color=black,dotted,very thick] coordinates {
    (15.4,0)
    (15.4,1.05)};
\end{axis}
\end{tikzpicture}
\caption{The payoff as a fraction of the face value, for different bonds, as a function of the value of the total assets $a$.}
\label{fig:1bank-debttype}
\end{subfigure}
\caption{Example~\ref{ex:setting-coco}: Illustrations of the value of CoCo bonds.}
\label{fig:1bank}
\end{figure}

For this bank, the conversion and default break points for the total asset variable are $a_1=10, a_2=11, a_3=15.4$ and the equity $E=E(a)$ is given explicitly by
\[E(a) = \begin{cases} a- 14 &\text{if } a \geq 15.4 \\ \frac{a}{11} &\text{if } a \in [11,15.4) \\ a - 10 &\text{if } a \in [10,11) \\ \frac{a}{2} - 10 &\text{if } a \in [0,10). \end{cases}\]
This outcome is plotted in Figure~\ref{fig:1bank-wealth}: It shows that in the region of CoCo conversions $a \in [11,15.4)$ the equity $E$ is made more stable as the external asset value declines.  
In contrast to the bank shareholders, who benefit uniformly from the CoCo bonds, holders of bank debt find that CoCo bonds underperform the payoff of the vanilla debt.  The value of the equity for the original shareholders is displayed in Figure~\ref{fig:1bank-equity} for three conversion factors $q = 0, 0.5, 1$, i.e., in which \$1 of CoCo debt is converted to $q$ of new equity at conversion.  If $q = 0$, the original shareholders retain the full equity of the bank; the value retained by the original shareholders drops as $q$ increases and more of the bank's equity is held by CoCo bond owners.  The CoCo payout, given as  the sum of remaining debt $(1-\lambda(a))\bar p^c$ and value of the equity $c(\lambda(a))E(a)^+$, is shown as a fraction of the face value in Figure~\ref{fig:1bank-debttype} for $q = 0.5,1$.  We see that the higher the conversion factor $q$ the more value an investor will recover from CoCo bonds; it is also true that the value of CoCo bonds will always be dominated by the payout of the vanilla debt.

\end{example}

\begin{remark}\label{rem:optionality}
\begin{enumerate}
  \item The example illustrates that although CoCo bonds appear to underperform vanilla debt whenever $a<15.4$, their existence stabilizes the health of the bank. A similar bank, financed with pure vanilla debt with face value $\bar p^0=14$, would default at $a = 14$, which demonstrates the value of CoCo bonds in a crisis scenario.

  \item Although we assume a mechanical rule for conversion of CoCo bonds from debt to equity as was done in, e.g.,~\cite{glasserman2012coco}, CoCo securities are often callable, i.e., have an optional structure in which the issuer has a right to convert the debt at the trigger level, but not the obligation to do so.  Provided that banks are equity maximizers and the conversion rate $q$ is at most 1, each institution will choose to exercise the CoCo option structure using the mechanical rules set out in this work at maturity of the claims.
\end{enumerate}
\end{remark}

\section{Contagion with Contingent Convertible Bonds}\label{sec:nbank}

The remainder of this work will focus on a network of $n$ financial institutions labelled by $i\in\{1,2,\dots, n\}$ whose initial vanilla and fractional CoCo debt obligations have face values $\bar p^0_i\ge 0$ and $\bar p^c_i\ge 0$ as in Section \ref{sec:1bank}.  Without loss of generality, we will assume $\bar p^0_i + \bar p^c_i > 0$ for all banks $i$. The CoCo parameters $\tau_i,q_i$ may vary across banks. These banks hold external assets and cross-holdings of interbank debt and equity. Like network models with vanilla debt and equity cross-holdings such as~\cite{EN01,RV13,gourieroux2012}, we seek to determine the vector of firm assets $A$ after network clearing at maturity $T = T_K$.  As in these papers, we assume the following stylized rules for clearing:\begin{enumerate}
\item \emph{Limited liabilities}: the total payment made by any firm will never exceed the total assets available to the bank.
\item \emph{Priority of debt claims}: a firm with equity $E_i<0$  cannot pay its debts in full and hence will default, in which case the shareholders receive no value.
\item \emph{All debts are of the same seniority}: in case a bank has $E_i<0$ and defaults, debts are paid out in proportion to the size of the nominal claims.
\end{enumerate}
Additionally, as in the prior section, a defaulted bank  will realize only a fixed fraction  $\alpha \in [0,1]$  of its total assets. As is common in the systemic risk literature, we also add a ``fictitious'' bank, labelled by $i=0$, to represent the external holders of bank debt and equity.

Our aim here is to determine the relationships between the balance sheets of all banks $i \in \{1,2,...,n\}$ just after clearing and CoCo conversion takes place at maturity $T=T_K$. We suppose the banks have the vector of external assets $X=(x_i)\in\bbr^n$, equities $E=(E_i)$ and conversion factors $\lambda\in[0,1]^n$.  The liabilities of bank $i$  consist of:
\begin{itemize}
    \item \emph{Vanilla  liabilities}: $L_{i0}^0:=\pi^0_{i0}\bar p^0_i \geq 0$ is owed from bank $i$ to  entities outside the financial system and $L_{ij}^0:=\pi^0_{ij}\bar p^0_i \ge 0$ is owed from bank $i$ to any other bank $j$;
    \item \emph{CoCo  liabilities}: $(1-\lambda_i)L_{i0}^c$  is owed from bank $i$ to entities outside the financial system where $L_{i0}^c:=\pi^c_{i0}\bar p^c_i \geq 0$  is the  initial face value, and $(1-\lambda_i)L_{ij}^c$  is owed from bank $i$ to any other bank $j$ where $L_{ij}^c:=\pi^c_{ij}\bar p^c_i \ge 0$.
      \item \emph{Total  liabilities}: $\bar p_i(\lambda_i):=\bar p^0_i+(1-\lambda_i)\bar p^c_i$ is the total amount owed by bank $i$.
      \end{itemize}
    Its nominal assets with equity cross-holdings are denoted by: 
\begin{itemize}
    \item \emph{External assets}: $x_i \geq 0$ is held in assets external to the financial network;
    \item \emph{Vanilla interbank debt assets}: $\sum_{j = 1}^n L_{ji}^0$ where bank $j$ owes $L_{ji}^0 \geq 0$ to bank $i$;

    \item \emph{CoCo interbank assets}: $\sum_{j = 1}^n (1 - \lambda_j) L_{ji}^c$ of remaining CoCo debt;
    \item \emph{Interbank equity assets}: bank $i$ holds a fraction $\pi_{ji}^e \in [0,1)$ of the original equity shares of bank $j$ plus the fraction $\pi^c_{ji}$ of the additional equity from the conversion of CoCo debt $\bar p^c_j$. 
\end{itemize}
We assume that $L_{ii}^0 =L_{ii}^c = 0$ to eliminate self-dealing and $\pi_{ii}^e = 0$ to eliminate double counting of a firm's equity.  Furthermore, note that $\sum_{j = 0}^n \pi_{ij}^e = 1$ since all equity is owned by some entity,  and the equity of bank $i$ is $E_i=A_i-\bar p_i(\lambda_i)$, its total realized assets minus its total liabilities. 

In the network setting,  the total asset value $A_i$ of a given bank $i$ is dependent on the values of securities written on the remaining banks $j\ne i$, which are themselves functions of their total asset values $A_j$ given in Section 2. Following the balance sheet construction for the valuation of assets, the total assets of bank $i$ are given by
 \begin{align}
\label{clearing_i} A_i &= x_i + \sum_{j=1}^n F_{ji}(A_j)\\
\label{clearing_F} \begin{split} F_{ji}(A_j) &:= \pi^0_{ji}[\bar p^0_j-E_j(A_j)^-]+\pi^c_{ji}[(1-\lambda_j(A_j))\bar p^c_j+ c_j(A_j) E_j(A_j)^+]\\
    &\qquad +\pi^e_{ji}(1-c_j(A_j)) E_j(A_j)^+\end{split}\\
\nonumber E_j(A_j) &:= \ind{A_j \geq a_{1,j}}A_j + \ind{A_j < a_{1,j}} \alpha A_j - [\bar p_j^0 + (1-\lambda_j(A_j))\bar p_j^c]\\
\nonumber \lambda_j(A_j) &:= 1 \wedge \left(\frac{\bar p_j^0 + \bar p_j^c - A_j/(1+\tau_j)}{\bar p_j^c}\right)^+\\
\nonumber c_j(A_j) &:= 1 - \left(\frac{(A_j \vee a_{2,j}) \wedge a_{3,j}}{a_{3,j}}\right)^{\frac{q_j}{\tau_j}}
\end{align}
with threshold asset levels: $a_{1,j} := \bar p_j^0;~a_{2,j} := (1+\tau_j)\bar p_j^0;~a_{3,j} := (1+\tau_j)(\bar p_j^0+\bar p_j^c)$.

\subsection{Existence of Clearing Wealths}\label{sec:nbank-main}
The network clearing problem to find solutions of \eqref{clearing_i} can now be characterized as finding solutions $A^*$ to the vector-valued fixed point equation:
\begin{align}
\label{eq_clearing} A^* &= \Phi(A^*,X)\ , \\
\nonumber  \Phi_i(A,X)&:= x_i + \sum_{j=1}^n F_{ji}(A_j).
\end{align}
This clearing condition can be compared with the clearing mechanism in, e.g., \cite{BF18comonotonic} for a vanilla interbank market with equity cross-holdings but without contingent convertible obligations. Just as in their setting, Tarski's fixed point theorem now implies that there exists a maximal and minimal clearing solution to this clearing problem. 

\begin{theorem}\label{thm:exist}
For any vector of external assets $X\in\bbr_+^n$, there exist maximal and minimal clearing total asset vector solutions $A^*_+(X),A^*_-(X)$ of  $A = \Phi(A,X)$ in 
\begin{equation}\label{eq:domain}
D := [\vec{0} \; , \; (I-\Pi^{e\T})^{-1}\left([X+\Pi^{0\T}\bar p^0+\Pi^{c\T}\bar p^c - \Pi^{e\T}[\bar p^0+\bar p^c]] \vee a_3\right)] \subseteq \bbr_+^n.
\end{equation}
\end{theorem}
\begin{proof} 
We will prove the existence of maximal and minimal clearing asset vectors through an application of Tarski's fixed point theorem.  Due to the construction of $\Phi$, this requires the monotonicity of $F_{ji}$ for every pair of banks $i,j$ and demonstrating that $\Phi(A,X) \in D$ for any $A \in D$.

First, consider the monotonicity of $F_{ji}$. We will prove this result by demonstrating separately that $\bar p_j^0-E_j(\cdot)^-$, $(1-\lambda_j(\cdot))\bar p_j^c + c_j(\cdot)E_j(\cdot)^+$, and $(1-c_j(\cdot))E_j(\cdot)^+$ are nondecreasing.
\begin{itemize}
\item Consider $\bar p_j^0-E_j(\cdot)^-$.  This is nondecreasing if $E_j(\cdot)$ is nondecreasing.  By expanding $\lambda_j$, $E_j$ is equivalently written as:
\[E_j(A_j) = \begin{cases} A_j - [\bar p_j^0+\bar p_j^c] &\text{if } A_j \geq a_{3,j} \\ \frac{\tau_j}{1+\tau_j}A_j &\text{if } A_j \in [a_{2,j},a_{3,j}) \\ A_j - \bar p_j^0 &\text{if } A_j \in [a_{1,j},a_{2,j}) \\ \alpha A_j - \bar p_j^0 &\text{if } A_j < a_{1,j}. \end{cases}\]  
By this formulation, $E_j$ is trivially nondecreasing (noting continuity at $a_{2,j},a_{3,j}$ and $\alpha \leq 1$).
\item Consider $(1-\lambda_j(\cdot))\bar p_j^c + c_j(\cdot)E_j(\cdot)^+$. Expanding all of these terms recovers: 
\[\begin{cases} \bar p_j^c &\text{if } A_j \geq a_{3,j} \\ \frac{1}{1+\tau_j}\left[(1+\tau_j)-(q_j+\tau_j)\left(\frac{A_j}{a_{3,j}}\right)^{\frac{q_j}{\tau_j}}\right] &\text{if } A_j \in [a_{2,j},a_{3,j}) \\ c_j(a_{2,j})[A_j - \bar p_j^0] &\text{if } A_j \in [a_{1,j},a_{2,j}) \\ 0 &\text{if } A_j < a_{1,j}. \end{cases} \]
By this formulation, monotonicity in $A_j$ is trivial (noting continuity throughout).
\item Consider $(1-c_j(\cdot))E_j(\cdot)^+$. First, as noted in the prior cases, $E_j$ is nondecreasing.  Furthermore, $1-c_j(A_j) = [((A_j\vee a_{2,j})\wedge a_{3,j})/a_{3,j}]^{q_j/\tau_j}$ is also trivially nondecreasing. 
\end{itemize}

Now we need to demonstrate that the lattice $D$ provides a consistent pre-image and image space for $\Phi(\cdot,X)$.  First, trivially, $\Phi(\vec{0},X) \geq \vec{0}$ by use of the Leontief inverse and $a_3 \in \bbr_+^n$.  Consider now $\Phi(\bigvee D,X)$. By construction $\bigvee D \geq a_3$ and, thus, $\lambda(\bigvee D) = \vec{0}$.  Using this value for $\lambda$, we find
\begin{align*}
\Phi&(\bigvee D,X)-\bigvee D = [X + \Pi^{0\T}\bar p^0 + \Pi^{c\T}\bar p^c - \Pi^{e\T}[\bar p^0+\bar p^c]] -(I-\Pi^{e\T})\bigvee D\\
&= [X + \Pi^{0\T}\bar p^0 + \Pi^{c\T}\bar p^c - \Pi^{e\T}[\bar p^0+\bar p^c]] - \left([X + \Pi^{0\T}\bar p^0 + \Pi^{c\T}\bar p^c - \Pi^{e\T}[\bar p^0+\bar p^c]] \vee a_3\right) \leq \vec{0}.
\end{align*}
Finally, by monotonicity of $\Phi(\cdot,X)$, it follows that $\Phi(A,X) \in D$ for any $A \in D$.
\end{proof}

\begin{remark}\label{rem:marketvaluatioin}
The maximal clearing vector guaranteed by Theorem \ref{thm:exist} determines fair market prices for all balance sheet entries immediately after clearing at $T_K$.  Under the assumptions of risk neutral (no-arbitrage) pricing theory, these fair market prices in turn determine market valuation of the banks' balance sheets at earlier times $t<T_K$. Thus the maximal clearing vector determines all market prices in a dynamic network pricing model. While exploration of such a dynamical model would be of great interest, the details are beyond the scope of the present paper. 
\end{remark}

\begin{remark}\label{rem:computation}
The maximal clearing solution $A^*_+(X) = \Phi(A^*_+(X),X)$ can be found via Picard iterations of $\Phi(\cdot,X)$ beginning at $A^{(0)} = \bigvee D$. 
By upper semicontinuity, these fixed point iterations will converge to the maximal clearing asset vector.
\end{remark}

\begin{remark}\label{rem:KV}
The multiple maturity interbank network setting of \cite{KV16}, and the CoCo systems with all-or-nothing conversion presented in~\cite{balter2020contingent} require additional monotonicity conditions to ensure existence of clearing solutions.  In our proof, the simple non-speculative condition $q\le 1$ is all that is needed to construct the monotone fixed point problem $A=\Phi(A,X)$, thereby guaranteeing the existence of maximal and minimal clearing solutions. 
\end{remark}

\begin{corollary}\label{cor:unique}
Consider the setting of Theorem~\ref{thm:exist} with no bankruptcy costs, i.e., $\alpha = 1$.  
If $L_{i0}^0 > 0$ and $\pi_{i0}^e > 0$ for every bank $i$ then there exists a \emph{unique} clearing wealth vector $A^* = \Phi(A^*,X)$.
\end{corollary}
\begin{proof}
Throughout this proof let $A_+^*,A_-^*$ be the maximal and minimal clearing wealths as determined by Theorem~\ref{thm:exist} with the parameter $X\in\bbr^n_+$ implicit.  Additionally let $\lambda_+^* = \lambda(A_+^*)$ and $\lambda_-^* = \lambda(A_-^*)$ be fixed; by monotonicity note that $\lambda_+^* \leq \lambda_-^*$ component-wise.

First, we will demonstrate that the societal wealth is the same under any clearing solution, i.e., $A_{+,0}^* = A_{-,0}^*$.
Let $A = \Phi(A,X)$ be an arbitrary clearing solution with associated fractional conversions $\lambda := \lambda(A)$.  Without loss of generality, we assume the societal node has $x_0 = 0$ and has no debt or equity obligations, so that 
\[A_0:=\sum_{j = 1}^n \left(\pi^0_{j0}(\bar p^0_j-E_j(A_j)^-)+(1-\lambda_j)\pi^c_{j0}\bar p^c_j + c_j(A_j)\pi^c_{j0}E_j(A_j)^+ + (1-c_j(A_j))\pi^e_{j0}E_j(A_j)^+\right)\ge 0.\]  
Recall that $\alpha=1$. The total equity in the system is:
\begin{align*}
\sum_{i = 0}^n E_i(A_i)^+ &= \sum_{i = 0}^n (E_i(A_i)+E_i(A_i)^-)\\
    &= \sum_{i = 1}^n \left(x_i  - \bar p_i(\lambda_i)+E_i(A_i)^-\right)\\
    &\quad + \sum_{j = 1}^n\sum_{i = 0}^n\left(\begin{array}{l}\pi^0_{ji}(\bar p^0_j-E_j(A_j)^-)+(1-\lambda_j)\pi^c_{ji}\bar p^c_j+c_j(A_j)\pi^c_{ji}E_j(A_j)^+\\ + (1-c_j(A_j))\pi^e_{ji}E_j(A_j)^+\end{array}\right)\\
    &= \sum_{i = 1}^n x_i + \sum_{i = 1}^n E_i(A_i)^+.
\end{align*}
The third equality follows since $\sum_{i = 0}^n \pi_{ji}^0 =\sum_{i = 0}^n \pi_{ji}^c= \sum_{i = 0}^n \pi_{ji}^e = 1$, and implies $A_0 = \sum_{i = 1}^n x_i$.  
Therefore, it must follow that the maximal and minimal clearing solutions coincide for the societal node (i.e., $A_{+,0}^* = A_{-,0}^*$).

Second, using the constant wealth of the societal node, we will prove by contradiction the uniqueness of the entire clearing wealth vector.  Assume there exists some bank $i$ such that $A_{+,i}^* > A_{-,i}^*$. Then
\begin{align*}
A_{+,0}^* &= \sum_{j = 1}^n \left(\begin{array}{l}\left[L_{j0}^0 + (1-\lambda_{+,j}^*)L_{j0}^c\right]\\ + \left[\ind{E_j(A_{+,j}^*) \geq 0}\left(c_j(A_{+,j}^*)\pi^c_{j0}+(1-c_j(A_{+,j}^*))\pi^e_{j0}\right) + \ind{E_j(A_{+,j}^*) < 0}\pi_{j0}^0\right] E_j(A_{+,j}^*)\end{array}\right)\\
    &> \sum_{j = 1}^n \left(\begin{array}{l}\left[L_{j0}^0 + (1-\lambda_{-,j}^*)L_{j0}^c\right]\\ + \left[\ind{E_j(A_{-,j}^*) \geq 0}\left(c_j(A_{-,j}^*)\pi_{j0}^c+(1-c_j(A_{-,j}^*))\pi_{j0}^e\right) + \ind{E_j(A_{-,j}^*) < 0}\pi_{j0}^0\right] E_j(A_{-,j}^*)\end{array}\right)\\ 
    &= A_{-,0}^*
\end{align*}
where the strict inequality follows from the same arguments as the proof of Theorem~\ref{thm:exist} (and the extra conditions guarantee strict monotonicity).  This contradiction completes our proof. 
\end{proof}

The next example illustrates that CoCo systems are essentially different from  ``vanilla'' Eisenberg-Noe clearing problems such as~\cite{gourieroux2012} which always exhibit piecewise linear relationships: Clearing wealths of CoCo systems are {\em piecewise non-linear} in the external asset vector $x$.  

\begin{example}\label{ex:symmetric}
An $n$ bank system will be called ``symmetric'' if: (i) $q,\tau,\alpha$ are the same for all banks; (ii) the total interbank liabilities in vanilla debt, CoCo debt and equity are the same for all banks; (iii) the external vanilla and CoCo liabilities are the same for all banks; (iv) the external assets, interbank CoCo assets, interbank vanilla assets, and interbank equity assets are the same for all banks. 

Let $\beta_0,\beta\in(0,1)$ be the fraction of total external and interbank debt that are CoCo-ized for all banks in a symmetric system, and let $x,y,z\ge 0$ denote the total external assets, external debt and interbank debt for each bank. Then, prior to conversion, each bank has identical vanilla debt $\bar p^0_i=(1-\beta_0)y+(1-\beta)z$ and CoCo debt $\bar p^c_i=\beta_0y+\beta z$. In addition, for a symmetric system, the interbank holding matrices $\Pi^0,\Pi^c,\Pi^e$ must have constant row and column sums, that is $\Pi^0\vec{1}=(\Pi^0)^\top\vec{1}=\pi^0\vec{1}$, $\Pi^c\vec{1}=(\Pi^c)^\top\vec{1}=\pi^c\vec{1}$, and $\Pi^e\vec{1}=(\Pi^e)^\top\vec{1}=\pi^e\vec{1}$ for scalars $\pi^0=\frac{(1-\beta)z}{(1-\beta_0)y+(1-\beta)z}$, $\pi^c=\frac{\beta z}{\beta_0 y + \beta z}$, and $\pi^e\in [0,1)$.

Under such conditions, one can verify that if an asset vector is symmetric (i.e., $A=a\vec{1}$ for a scalar $a$) then $\Phi(A,x\vec{1})=\phi(a,x)\vec{1}$ for a scalar function $\phi(a,x)$. Moreover $\phi(a,x)=x+F(a)$ where $F$ is defined comparably to \eqref{clearing_F}, i.e.,
\begin{align*}
F(a) &= (1-\lambda(a)\beta)z+[c(a)\pi^c + (1-c(a))\pi^e]E(a)^+ -\pi^0 E(a)^-\\
E(a) &= \ind{a \geq a_1}a + \ind{a < a_1} \alpha a - [(1-\lambda(a)\beta_0)y+(1-\lambda(a)\beta)z]\\
\lambda(a) &= 1 \wedge \left(\frac{y+z - a/(1+\tau)}{\beta_0 y + \beta z}\right)^+\\
c(a) &= 1-\left(\frac{(a\vee a_2)\wedge a_3}{a_3}\right)^{q/\tau}\\
\end{align*} 
with intervals $[0,a_1),[a_1,a_2),[a_2,a_3),[a_3,\infty)$ as in Section~\ref{sec:1bank} for
\begin{align*}
a_1 &= (1-\beta_0)y+(1-\beta)z\\
a_2 &= (1+\tau)[(1-\beta_0)y+(1-\beta)z]\\
a_3 &= (1+\tau)[y+z].
\end{align*}
Thus the clearing wealth condition boils down to the scalar fixed point condition $a^*=\phi(a^*,x)$. Moreover, since the hyperinterval $D$ given by \eqref{eq:domain} is symmetric, the maximal and minimal clearing solutions, obtained by iterating the function starting from either the maximal or minimal points of $D$, must both be symmetric.

In fact, the intervals in total assets determined by the thresholds $a_1,a_2,a_3$ can be mapped into intervals in the external assets for the maximal clearing solution $a^*(x)$, i.e., $[0,x_1),[x_1,x_2),[x_2,x_3),[x_3,\infty)$ where 
\begin{align*}
x_1 &= (1-\beta_0)y, \\
x_2 &= (1-\beta_0)y + \tau\left(1-[c(a_2)\pi^c+(1-c(a_2))\pi^e]\right)[(1-\beta)z+(1-\beta_0)y],\\
x_3 &= y + \tau(1-\pi^e)(y+z)  
\end{align*}
On the most challenging interval, namely $x\in[x_2,x_3)$, one needs to solve the equation $G(a)-(1+\tau) (x + z - \pi^c(y+z)) = 0$ on $a\in[a_2,a_3)$ for 
\[G(a):= (1+\tau)(1-\pi^c)a + \tau(\pi^c-\pi^e)\left(\frac{a}{a_3}\right)^{\frac{q}{\tau}}a \ .\]
Since $G(a_2)-(1+\tau)(x+z-\pi^c(y+z))=(1+\tau)(x_2-x) \leq 0 < (1+\tau)(x_3-x)=G(a_3)-(1+\tau) (x+z-\pi^c(y+z))$ and $G'(a) = 1 - \pi^c + q(\pi^c-\pi^e)(1-c(a))+\tau(1-[c(a)\pi^c+(1-c(a))\pi^e]) > 0$ for every $a \in [a_2,a_3)$, there is a unique root $a^{\ast\ast}(x)\in[a_2,a_3)$ for any $x\in[x_2,x_3)$. For the  intervals $[0,x_1), [x_1,x_2)$, and $[x_3,\infty)$, $a^*(x)$ is explicitly computable, showing that for any external asset level $x\ge 0$ the maximal clearing vector of the symmetric system has the form $A_+=a_+\vec{1}$ where
\begin{equation*}
a_+ =a^*(x):= \begin{cases} \frac{(x+z-\pi^e[y+z]}{1 - \pi^e} &\text{if } x\ge x_3 \\
   a^{\ast\ast}(x) &\text{if } x_2\le x< x_3 \\
    \frac{x+(1-\beta)z - [c(a_2)\pi^c+(1-c(a_2))\pi^e][(1-\beta_0)y+(1-\beta)z]}{1 - [c(a_2)\pi^c+(1-c(a_2))\pi^e]} &\text{if } x_1\le x<x_2  \\
   \frac{[(1-\beta)z+(1-\beta_0)y]x}{(1-\alpha)(1-\beta)z + (1-\beta_0)y} &\text{if } 0\le x<x_1\end{cases}
\end{equation*}

Interestingly, one can also show that the maximal clearing vector is unique except when $x\in[x_1,x_0), x_0:=\alpha^{-1}[(1-\alpha)(1-\beta)z + (1-\beta_0)y]$. If $x\in[x_1,x_0),$  the minimal clearing vector has the form $A_-=a_-\vec{1}$ with
\begin{align*}
a_- &= \frac{[(1-\beta)z+(1-\beta_0)y]x}{(1-\alpha)(1-\beta)z + (1-\beta_0)y}\ .
\end{align*}

This symmetric system is numerically compared with Example~\ref{ex:setting-coco} in Figure~\ref{fig:symmetric} to demonstrate the impacts of the network effects, especially due to CoCo bonds, on the financial system.  Specifically, in this example system, we find that, so long as the (symmetric) banks are not in default, both the banks and the original external shareholders benefit from the network effects.  In contrast, the external bond holders are marginally harmed by the introduction of network effects.  Notably, also, the CoCo bonds are triggered at higher values of the external asset with network effects, but are not triggered in full until the external asset is at a lower value; further, default occurs at the same threshold value $x_1$ without regard to the introduction of the financial network.
\begin{figure}[t]
\centering
\begin{subfigure}[t]{0.315\textwidth}
\centering
\begin{tikzpicture}[declare function={t = 0.1; eta = 0.5; b0 = 4/14; b = 3/4; Y = 14; Z = 10; piE0 = 0.2; alpha = 1/2; 
    piC = b*Z/(b*Z+b0*Y);
    pbar(\l) = (1-b*\l)*Z + (1-b0*\l)*Y; 
    c(\l) = (1 - (pbar(\l)/pbar(0))^(eta/t)); piE(\l) = c(\l)*piC + (1-c(\l))*piE0;
    G(\v) = (1+t)*(1-piC)*\v+t*(piC-piE0)*(\v/(t*pbar(0)))^(eta/t)*\v;},]
\begin{axis}[
    width=\textwidth,
    xlabel={External Asset $x$},
    ylabel={Wealth $E(a_+)$},
    xmin=5, xmax=20,
    ymin=-8, ymax=8,
    label style={font=\scriptsize},
    ticklabel style={font=\tiny},
    legend pos=north west,
    legend style={nodes={scale=0.4, transform shape}},
    ]
\def\xA{(1-b0)*Y}
\def\xB{(1-b0)*Y + t*(1-piE(1))*pbar(1)}
\def\xC{Y + t*(1-piE(0))*pbar(0)}
\def\vA{t*pbar(1)}
\def\vB{t*pbar(0)}

\addplot [domain=\xC:20,solid,color=blue,very thick]{(x-Y)/(1-piE(0))};
\draw [domain=\vA:\vB,solid,color=blue,very thick,variable=\v] plot ({G(\v)/t-(Z-piC*pbar(0))}, {\v});
\addplot [domain=\xA:\xB,solid,color=blue,very thick,forget plot]{(x-(1-b0)*Y)/(1-piE(1))};
\addplot [domain=5:\xA,solid,color=blue,very thick,forget plot]{(pbar(1)*(alpha*x-pbar(1)+alpha*(1-b)*Z))/(pbar(1)-alpha*(1-b)*Z)};
\addlegendentry{With network}

\addplot [domain=5:10,dashed,color=red,very thick]{x/2-10};
\addplot [domain=10:11,dashed,color=red,very thick,forget plot]{x-10};
\addplot [domain=11:15.4,dashed,color=red,very thick,forget plot]{x/11};
\addplot [domain=15.4:20,dashed,color=red,very thick,forget plot]{x-14};
\addlegendentry{Without network}

\addplot [color=black,dotted,very thick] coordinates {
    (\xA,-8)
    (\xA,8)};
\addplot [color=black,dotted,very thick] coordinates {
    (\xB,-8)
    (\xB,8)};
\addplot [color=black,dotted,very thick] coordinates {
    (\xC,-8)
    (\xC,8)};
\addplot [color=gray,dashed,very thick] coordinates {
    (5,0)
    (20,0)};
\end{axis}
\end{tikzpicture}
\caption{The maximal symmetric clearing wealth $E(a_+)$ as a function of the value of the external assets $x$.}
\label{fig:symmetric-wealth}
\end{subfigure}
~
\begin{subfigure}[t]{0.315\textwidth}
\centering
\begin{tikzpicture}[declare function={t = 0.1; eta = 0.5; b0 = 4/14; b = 3/4; Y = 14; Z = 10; piE0 = 0.2; alpha = 1/2; 
    piC = b*Z/(b*Z+b0*Y);
    pbar(\l) = (1-b*\l)*Z + (1-b0*\l)*Y; 
    c(\l) = (1 - (pbar(\l)/pbar(0))^(eta/t)); piE(\l) = c(\l)*piC + (1-c(\l))*piE0;
    lam(\v) = (t*pbar(0)-\v)/(t*(pbar(0)-pbar(1)));
    G(\v) = (1+t)*(1-piC)*\v+t*(piC-piE0)*(\v/(t*pbar(0)))^(eta/t)*\v;
    lam0(\x) = (1.4-\x/11)/0.4; c0(\l) = (1-(1-2/7*\l)^5);},]
\begin{axis}[
    width=\textwidth,
    xlabel={External Asset $x$},
    ylabel={Value of Equity $E(a_+)^+$},
    xmin=5, xmax=20,
    ymin=0, ymax=3,
    label style={font=\scriptsize},
    ticklabel style={font=\tiny},
    legend pos=north west,
    legend style={nodes={scale=0.4, transform shape}},
    ]
\def\xA{(1-b0)*Y}
\def\xB{(1-b0)*Y + t*(1-piE(1))*pbar(1)}
\def\xC{Y + t*(1-piE(0))*pbar(0)}
\def\vA{t*pbar(1)}
\def\vB{t*pbar(0)}

\addplot [domain=\xC:20,solid,color=blue,very thick]{x-Y};
\draw [domain=\vA:\vB,solid,color=blue,very thick,variable=\v] plot ({G(\v)/t-(Z-piC*pbar(0))} , {G(\v)/t-(Z-piC*pbar(0)) - ((1-b0*lam(\v))*Y+c(lam(\v))*(1-piC)*\v)});
\draw [domain=0:\vA,solid,color=blue,very thick,variable=\v] plot ({(1-b0)*Y+(1-piE(1))*\v} , {(1-piE(1))*\v - c(1)*(1-piC)*\v});
\addplot [domain=5:\xA,solid,color=blue,very thick,forget plot]{0};
\addlegendentry{With network}

\addplot [domain=15.4:20,dashed,color=red,very thick]{x-14};
\addplot [domain=11:15.4,dashed,color=red,very thick,forget plot]{(1-c0(lam0(x)))*x/11};
\addplot [domain=10:11,dashed,color=red,very thick,forget plot]{(1-c0(1))*(x-10)};
\addplot [domain=5:10,dashed,color=red,very thick,forget plot]{0};
\addlegendentry{Without network}

\addplot [color=black,dotted,very thick] coordinates {
    (\xA,0)
    (\xA,6)};
\addplot [color=black,dotted,very thick] coordinates {
    (\xB,0)
    (\xB,6)};
\addplot [color=black,dotted,very thick] coordinates {
    (\xC,0)
    (\xC,6)};
\end{axis}
\end{tikzpicture}
\caption{The value of a single bank's equity owned by the original \emph{external} shareholders as a function of the value of the external assets $x$.}
\label{fig:symmetric-equity0}
\end{subfigure}
~
\begin{subfigure}[t]{0.315\textwidth}
\centering
\begin{tikzpicture}[declare function={t = 0.1; eta = 0.5; b0 = 4/14; b = 3/4; Y = 14; Z = 10; piE0 = 0.2; alpha = 1/2; 
    piC = b*Z/(b*Z+b0*Y);
    pbar(\l) = (1-b*\l)*Z + (1-b0*\l)*Y; 
    c(\l) = (1 - (pbar(\l)/pbar(0))^(eta/t)); piE(\l) = c(\l)*piC + (1-c(\l))*piE0;
    lam(\v) = (t*pbar(0)-\v)/(t*(pbar(0)-pbar(1)));
    G(\v) = (1+t)*(1-piC)*\v+t*(piC-piE0)*(\v/(t*pbar(0)))^(eta/t)*\v;
    lam0(\x) = (1.4-\x/11)/0.4; c0(\l) = (1-(1-2/7*\l)^5);},]
\begin{axis}[
    width=\textwidth,
    xlabel={External Asset $x$},
    ylabel={Value of Debt $D(a_+)$},
    xmin=5, xmax=20,
    ymin=0.7, ymax=1.02,
    label style={font=\scriptsize},
    ticklabel style={font=\tiny},
    legend pos=north west,
    legend style={nodes={scale=0.4, transform shape}},
    ]
\def\xA{(1-b0)*Y}
\def\xB{(1-b0)*Y + t*(1-piE(1))*pbar(1)}
\def\xC{Y + t*(1-piE(0))*pbar(0)}
\def\vA{t*pbar(1)}
\def\vB{t*pbar(0)}

\addplot [domain=\xC:20,solid,color=blue,very thick]{1};
\draw [domain=\vA:\vB,solid,color=blue,very thick,variable=\v] plot ({G(\v)/t-(Z-piC*pbar(0))} , {(1-b0*lam(\v))+c(lam(\v))*(1-piC)*\v/Y});
\draw [domain=0:\vA,solid,color=blue,very thick,variable=\v] plot ({(1-b0)*Y+(1-piE(1))*\v} , {(1-b0)+c(1)*(1-piC)*\v/Y});
\addplot [domain=5:\xA,solid,color=blue,very thick,forget plot]{alpha*(1-b0)*x/(pbar(1)-alpha*(1-b)*Z)};
\addlegendentry{With network}

\addplot [domain=5:10,dashed,color=red,very thick]{alpha*x/14};
\addplot [domain=10:11,dashed,color=red,very thick,forget plot]{10/14+4/14*c0(1)*(x-10)/4};
\addplot [domain=11:15.4,dashed,color=red,very thick,forget plot]{10/14+4/14*((1-lam0(x))*4 + c0(lam0(x))*x/11)/4};
\addplot [domain=15.4:20,dashed,color=red,very thick,forget plot]{1};
\addlegendentry{Without network}

\addplot [color=black,dotted,very thick] coordinates {
    (\xA,0)
    (\xA,1.05)};
\addplot [color=black,dotted,very thick] coordinates {
    (\xB,0)
    (\xB,1.05)};
\addplot [color=black,dotted,very thick] coordinates {
    (\xC,0)
    (\xC,1.05)};
\end{axis}
\end{tikzpicture}
\caption{The payoff to the external node as a fraction of the face value as a function of the value of the external assets $x$.}
\label{fig:symmetric-debt0}
\end{subfigure}
\caption{Example~\ref{ex:symmetric}: Illustrations of the impact of a symmetric interbank network.  
The CoCo parameters are as in Example~\ref{ex:setting-coco} with $q = \frac{1}{2}$ so that the red dashed curves describe the system without an interbank network.
The networked system is such that $y = 14$ with $\beta_0 = \frac{4}{14}$, $z = 10$ with $\beta = \frac{3}{4}$, and interbank equity holdings $\pi^e = \frac{1}{5}$.}
\label{fig:symmetric}
\end{figure}
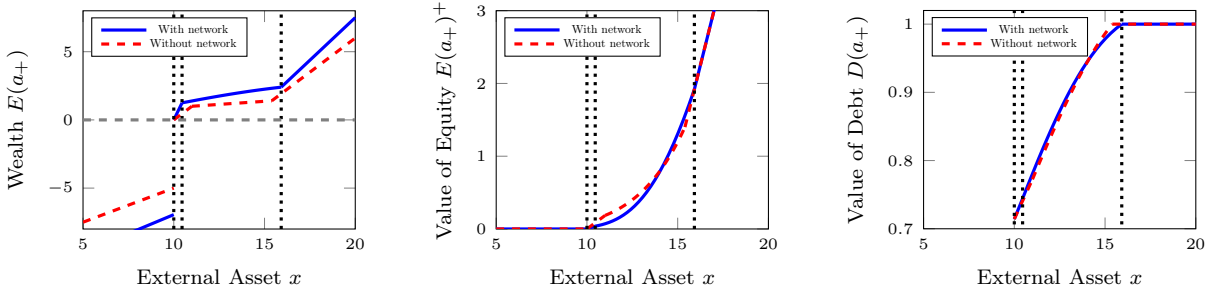

\end{example}

\section{Do CoCos Reduce Systemic Risk?}\label{sec:discussion}

The general interbank clearing problem of Section~\ref{sec:nbank} reduces to the model of equity cross-holdings and  vanilla debt   of \cite{RV13} when CoCo-ization is removed. It is of great interest to ask if systemic risk measures decrease as CoCo-ization is increased. This section will provide support for such a conclusion in certain ``typical systems'' when financial stability is measured by the number of defaulted banks, while also providing a simple counterexample to show that it is not universally true. 

First we note a very simple setting in which CoCo bonds improve financial stability as measured by the number of defaulted banks. 

\begin{proposition}\label{prop:full-coco}
A contingent system of obligations that is totally CoCo-ized, i.e., it has $\bar p^0 = \vec{0}$, will have no defaults.  
\end{proposition}
\begin{proof}
By construction, bank $i$ defaults if and only if $A_i^* < a_{1,i} := \bar p_i^0 = 0$. Since the assets $A_i^* \geq 0$, this can never occur.  
This contradiction implies that no bank may default. 
\end{proof}

\begin{assumption}\label{ass:maximal}
For the remainder of this work, all results and examples will be presented solely w.r.t.\ the maximal clearing assets $A^*_+$.
\end{assumption}

\subsection{Strongly Symmetric Systems}\label{sec:discussion-symmetric}

Here we consider an $n$ bank symmetric system as in Example~\ref{ex:symmetric} with liability parameters $\beta,\beta_0,y,z$.
In this subsection we stress banks by differing amounts, and to that end we allow the external assets $x_i$ to be non-equal. We also impose stronger symmetry conditions that all off-diagonal matrix entries do not depend on $i\ne j$:
\[\pi^0_{ij}=(n-1)^{-1}\pi^0,\quad \pi^c_{ij}=(n-1)^{-1}\pi^c,\quad \pi^e_{ij}=(n-1)^{-1}\pi^e\]
 where $\pi^0:=\frac{(1-\beta)z}{(1-\beta_0)y+(1-\beta)z},~\pi^c:=\frac{\beta z}{\beta_0 y + \beta z}$.
Note that this is a {\em complete financial network} in the sense of \cite{AOT15} (without CoCo obligations) and~\cite{calice2020contingent} (with CoCo obligations).
Under these assumptions, the clearing equations for $A$ can be written as
\begin{equation}
\label{eq:cc-a}
A_i=x_i+\sum_{j=1}^n F(A_j) - F(A_i), \quad i=1,2,\dots, n \ ,
\end{equation}
for any vector $X=(x_i)_{i=1,\dots, n}\in \bbr_+^n$. Following the notation of Example~\ref{ex:symmetric}, we note that the univariate functions $F_{ij},E_i,\lambda_i,c_i$ do not depend on $i\ne j$, i.e.,
\begin{align*}
F_{ij}(a) &= F(a) = \frac{1}{n-1}[(1-\lambda(a)\beta)z - \pi^0 E(a)^- + (c(a)\pi^c + (1-c(a))\pi^e)E(a)^+] \\
E_i(a) &= E(a) = \ind{a \geq a_1}a + \ind{a < a_1}\alpha a - [(1-\lambda(a)\beta_0)y + (1-\lambda(a)\beta)z] \\
\lambda_i(a) &= \lambda(a) = 1 \wedge \left(\frac{y+z - a/(1+\tau)}{\beta_0 y + \beta z}\right)^+ \\
c_i(a) &= c(a) = 1 - \left(\frac{(a \vee a_2) \wedge a_3}{a_3}\right)^{\frac{q}{\tau}} \ 
\end{align*}
with threshold asset levels $a_1 = (1-\beta_0)y + (1-\beta)z$, $a_2 = (1+\tau)[(1-\beta_0)y + (1-\beta)z]$, and $a_3 = (1+\tau)[y+z]$ which also do not depend on the bank.

Symmetric systems such as this have the following strong monotonicity property.
\begin{lemma}\label{lemma:cc-mono} 
The clearing vector $A$ in a strongly symmetric system with external asset vector $X \in \bbr_+^n$ satisfies the following monotonicity condition: for any $i\ne j$, $x_i<x_j$  implies $A_i \leq A_j$ with strict inequality if $A_j \geq a_1$ or $\alpha > 0$. 
\end{lemma}
\begin{proof}
If $x_i<x_j$ we  suppose that  $A_i\ge A_j$ and derive a contradiction. 
Assume bank $i$ is solvent, that is $A_i \geq a_1$, the clearing equation \eqref{eq:cc-a} gives the contradiction
\begin{align*}
A_i &= x_i + \sum_{k = 1}^n F(A_k) - F(A_i) < x_j + \sum_{k = 1}^n F(A_k) - F(A_j) = A_j\ .
\end{align*}
Similarly, the clearing equations provide a contradiction if bank $i$ is insolvent, that is $A_i < a_1$, provided that either bank $j$ is solvent ($A_j \geq a_1$) or with some recovery in case of default ($\alpha > 0$).  Notably if $x_i < x_j$ but with both banks defaulting and no recovery in case of default, it is immediate to see $A_i = x_i + \sum_{k = 1}^n F(A_k) - F(A_i) = A_j$ because $\lambda(A_i) = \lambda(A_j)$ and $E(A_i) = E(A_j)$.
\end{proof}

We consider a fully solvent strongly symmetric system with $X=x\vec{1}, x\ge (1-\beta_0)y$ and apply an equal stress  $\epsilon>0$ to the external assets of a subset of $d<n$ banks. Lemma  \ref{lemma:cc-mono} then suggests that there will typically be two critical values $\epsilon_1<\epsilon_2$ such that no banks default if $\epsilon\le\epsilon_1$, the stressed banks only default if $\epsilon_1< \epsilon\le \epsilon_2$, and all banks default if $\epsilon> \epsilon_2$. The following theorem gives the important result that the critical values $\epsilon_1$ and $\epsilon_2$ (if they exist) are nondecreasing as functions of the two variables $\beta,\beta_0$. This implies the desired general result that the number of defaulted banks in the stressed system is nonincreasing in $\beta,\beta_0$. 

\begin{theorem}\label{thm:cc-coco}
Consider a strongly symmetric system of $n$ banks with $x \geq (1-\beta_0)y$ and either $\pi^e \leq \frac{\beta z}{\beta_0 y + \beta z}$ or $\tau \in (0,q]$. For any $d<n$, if the external assets of $d$ banks are subjected to shocks of size $\epsilon$, then the critical values $ \epsilon_1\le \epsilon_2$ are nondecreasing functions of $\beta,\beta_0$.
\end{theorem}
\begin{proof}[Sketch proof] 
For the purposes of this proof, we will consider an equivalent formulation of the clearing problem~\eqref{eq_clearing} which simplifies the computations herein.  Consider $V_i = E_i(A_i)$ to be the equity for any bank $i$ where $A$ is any clearing assets vector.  It can be observed that $V = E(A)$ if and only if
\begin{align*}
V_i &= E_i(x_i + \sum_{j = 1}^n A_{ji}(V_j)), \quad \forall i = 1,...,n \\
A_{ji}(V_j) &= \pi^0_{ji}(\bar p^0_j-V_j^-)+\pi^c_{ji}\left[(1-\lambda_j(V_j))\bar p^c_j + c_j(V_j)V_j^+\right] +\pi^e_{ji}(1-c_j(V_j))V_j^+ \\
\lambda_j(V_j) &= 1 \wedge \left(\frac{\tau_j[\bar p_j^0+\bar p_j^c] - V_j}{\tau_j \bar p_j^c}\right)^+ \\
c_j(V_j) &= 1 - \left(\frac{\bar p_j^0 + (1-\lambda_j(V_j))\bar p_j^c}{\bar p_j^0 + \bar p_j^c}\right)^{\frac{q_j}{\tau_j}}
\end{align*}
with the abuse of notation for $\lambda,c$.
Furthermore, by the monotonicity of $E$ as provided in the proof of Theorem~\ref{thm:exist}, we consider the maximal clearing equity $V$ throughout this proof which corresponds with the maximal clearing assets $A$.
With this modification, bank $i$ defaults if and only if $V_i < 0$; this is equivalent to the prior condition that $A_i < a_{1,i}$.

Within this proof, due to the assumption of a strongly symmetric system, we can simplify these constructions so that:
\begin{align*}
A_{ij}(v) &= A(v) = \frac{1}{n-1}\left[(1 - \lambda(v)\beta)z - \pi^0 v^- + [c(v)\pi^c + (1-c(v))\pi^e]v^+\right] \\
\lambda_i(v) &= \lambda(v) = 1 \wedge \left(\frac{\tau[y+z] - v}{\tau[\beta_0 y + \beta z]}\right)^+ \\
c_i(v) &= c(v) = 1 - \left(\frac{(1-\lambda(v)\beta_0)y + (1-\lambda(v)\beta)z}{y+z}\right)^{\frac{q}{\tau}}\ .
\end{align*}

Let $d<n$ and define $v_s,v_n$ to be the clearing wealths of stressed (s) and nonstressed (n) banks. We provide here only a sketch of the proof for $\epsilon_2$ under the additional assumption $\alpha<1$. We can also assume $x < (1-\beta_0)y + (1-\beta)z$, since otherwise, $v_n > 0$ for any $\epsilon>0$ and hence $\epsilon_2$ does not exist. This case provides the logic utilized for $\epsilon_1$ as well but with fewer technicalities.  Full details of the proof are given in the Supplemental Material.

Lemma \ref{lemma:cc-mono}, and monotonicity of $E$ as provided in the proof of Theorem~\ref{thm:exist}, implies that $v_s<v_n$ for all $\epsilon \in (0,\epsilon_2]$. Further, when $\epsilon=\epsilon_2$, the clearing solution can be defined by the two equations
\begin{align*}
v_s &= \alpha(x-\epsilon_2 +(d-1) A(v_s)+(n-d) A(v_n))-\bar p(\lambda(v_s)) < 0, \\
v_n &=  x +d A(v_s)+(n-d-1)A(v_n)-\bar p(\lambda(v_n))=0
\end{align*}
Under these conditions, $\lambda(v_s)=\lambda(v_n)=1$, $A(v_s)=\frac{(1-\beta)z}{n-1}\left[1 + \frac{1}{(1-\beta_0)y+(1-\beta)z}v_s\right]$, $A(v_n)=\frac{(1-\beta)z}{n-1}$, and $\bar p(\lambda(v_n))= \bar p(\lambda(v_s))=(1-\beta_0)y+(1-\beta)z$. We find reduced equations for $v_s=v_s(\beta,\beta_0)$:
\begin{align*}
v_s &= \alpha\left(x-\epsilon_2 +\frac{d-1}{n-1} \frac{(1-\beta)z}{(1-\beta_0)y + (1-\beta)z} v_s\right) -(1-\beta_0)y-(1-\alpha)(1-\beta)z \\
0 &= x+\frac{d}{n-1} \frac{(1-\beta)z}{(1-\beta_0)y + (1-\beta)z} v_s -(1-\beta_0)y 
\end{align*}
The gradient of the second equation gives 
\begin{align*}
\frac{\partial v_s}{\partial \beta} &= -\frac{n-1}{d}\frac{(1-\beta_0)y}{(1-\beta)^2 z} (x - (1-\beta_0)y), \\
\frac{\partial v_s}{\partial \beta_0} &= -\frac{(n-1)}{d} \frac{y}{(1-\beta)z}\left((1-\beta_0)y + (1-\beta)z - (x - (1-\beta_0)y)\right)\ .
\end{align*}
Taking the gradient of the first equation leads to the desired equation 
\begin{align*}
\alpha \frac{\partial \epsilon_2}{\partial \beta} &= (1-\alpha)z - \frac{\partial v_s}{\partial \beta} \geq 0, \\
\alpha \frac{\partial \epsilon_2}{\partial \beta_0} &= (1 - \frac{d-1}{d}\alpha)y - \frac{\partial v_s}{\partial \beta_0} \geq 0
\end{align*}
which shows $(\partial \epsilon_2/\partial \beta, \partial \epsilon_2/\partial \beta_0) \in\bbr^2_+$.
\end{proof}

\begin{remark}\label{rem:acemoglu}
 Theorem~\ref{thm:cc-coco} applies for shocks to an arbitrary number $d$ of banks, and thus the result for financial stability holds for random shocks as studied in~\cite{AOT15}.  That is, under correlated Bernoulli shocks all of equal size $\epsilon > 0$, both the probability of an arbitrary bank defaulting and the expected number of defaulting banks decreases as CoCo-ization increases.
\end{remark}

\subsection{Monotonicity: a Counterexample}
\label{sec:discussion-counterexample}


Although we have demonstrated that stability of strongly symmetric systems is always improved by increasing interbank and external CoCo-ization, the following simple counterexample shows that CoCo bonds do not improve financial stability in all networks.
\begin{example}\label{ex:counterexample}
\begin{figure}[t]
\centering
\begin{tikzpicture}
\tikzset{node style/.style={state, minimum width=0.36in, line width=0.5mm, align=center}}
\node[node style] at (0,0) (x1) {$x_1$};
\node[node style] at (5,0) (x2) {$x_2$};
\node[node style] at (10,0) (x0) {$x_0$};

\draw[every loop, auto=right, line width=0.5mm, >=latex]
(x1) edge[bend left=20] node[above] {$L_{12} = \bar p_1$} (x2)
(x2) edge[bend left=20] node[below] {$L_{21}$} (x1)
(x2) edge node[above] {$L_{20}$} (x0);
\end{tikzpicture}
\caption{Example~\ref{ex:counterexample}: Network topology of interbank obligations.}
\label{fig:counterexample-system}
\end{figure}
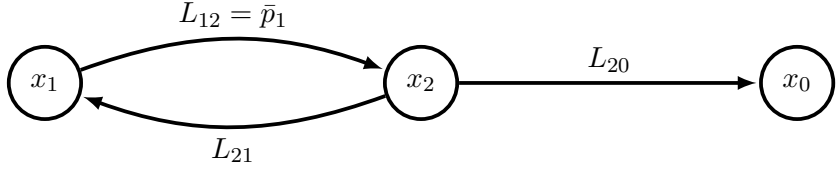
We consider the 2 bank and societal node system depicted in Figure~\ref{fig:counterexample-system} where some obligations have been CoCo-ized.  Throughout, we will specify this system with no vanilla interbank equity assets (i.e., $\pi_{ij}^e = 0$ for $i,j \in \{1,2\}$).  Let the external assets be $x_1 = 6$ and $x_2 = 1$, the interbank liabilities be $L_{12} = 10$, $L_{20} = L_{21} = 5$, and all other liabilities be identically $0$. Assuming zero recovery at default $\alpha = 0$, it can readily be verified that maximal clearing assets $A^*$ follow
\[A_1^* = \begin{cases} 11 &\text{if } \beta_1 \in [0 \; , \; \frac{9-\sqrt{41}}{20}] \\ 6 &\text{if } \beta_1 \in (\frac{9-\sqrt{41}}{20} \; , \; \frac{4}{10}) \\ 6 &\text{if } \beta_1 \in [\frac{4}{10} \; , \; \frac{7}{10}) \\ 6 &\text{if } \beta_1 \in [\frac{7}{10} \; , \; 1] \end{cases} 
\quad \text{ and } \quad 
A_2^* = \begin{cases} 10\beta_1^2 - 9\beta_1 + 11 &\text{if } \beta_1 \in [0 \; , \; \frac{9-\sqrt{41}}{20}] \\ 1 &\text{if } \beta_1 \in (\frac{9-\sqrt{41}}{20} \; , \; \frac{4}{10}) \\ 10\beta_1^2 - 14\beta + 11 &\text{if } \beta_1 \in [\frac{4}{10} \; , \; \frac{7}{10}) \\ -3\beta_1 + \frac{33}{10} &\text{if } \beta_1 \in [\frac{7}{10} \; , \; 1] \end{cases}\]
when $\beta_1 \in [0,1]$ of bank $1$'s liabilities are CoCo-ized with trigger level $\tau_1 = 1$ and with full conversion $q_1 = 1$. 
These analytical results for the CoCo-ized system can be directly compared with the vanilla interbank network ($\beta_1 = 0$). Specifically, bank $1$'s wealth (and the total (system-wide) wealth) improves with increased CoCo-ization $\beta_1$ for $\beta_1 \in [0,\frac{9-\sqrt{41}}{20}]$. However, this monotonicity does \emph{not} hold generally.  This is made more explicit by studying the set of defaulting banks (i.e., where $A_1^* < 10(1-\beta_1)$ or $A_2^* < 10$) as a function of the CoCo-ization level $\beta_1$:
\[\begin{cases} \emptyset &\text{if } \beta_1 \in [0 \; , \; \frac{9-\sqrt{41}}{20}] \\ \{1,2\} &\text{if } \beta_1 \in (\frac{9-\sqrt{41}}{20} \; , \; \frac{4}{10}) \\ \{2\} &\text{if } \beta_1 \in [\frac{4}{10} \; , \; \frac{7}{10}) \\ \{2\} &\text{if } \beta_1 \in [\frac{7}{10} \; , \; 1]. \end{cases}\]
As can be observed from this set of defaulting institutions, the introduction of CoCo bonds can actually make the financial system worse.  Furthermore, the impacts of CoCo-ization are non-trivial and non-monotonic. 

Before continuing, we want to understand which aspects of this network cause the non-monotonicity of system health w.r.t.\ the CoCo-ization level $\beta_1$ of bank $1$.  To understand this system behavior, consider the sequence of events generated from the Picard iterations as described in Remark~\ref{rem:computation}.  First, due to the stresses within this simple system, bank $1$ will convert some CoCo debts to equity regardless of bank $2$'s wealth. This CoCo conversion results in an immediate markdown in bank $2$'s assets. The more CoCo debt to be converted, the larger the markdowns; once bank $1$ converts enough debt to equity, bank $2$ can no longer remain solvent which triggers further deterioration in system health. The key aspect to this worsening system performance is the heterogeneity of the system. If, as in the strongly symmetric systems of Section~\ref{sec:discussion-symmetric}, both bank $1$ and bank $2$ were CoCo-ized at comparable levels, then the CoCo debts of bank $2$ would absorb the markdowns caused by bank $1$'s CoCo conversions. This is displayed in Figure~\ref{fig:counterexample} where it can be seen that no defaults occur along the $\beta_1 = \beta_2$ line.

We now consider more broadly how the trigger levels $\tau \in \bbr^2_{+}$ and fraction of liabilities made contingent $\beta \in [0,1]^2$ interact in this example.  To do so we specify three trigger levels: low, medium, and high ($\tau_1 = \tau_2 \in \{1,2,5\}$).  These three scenarios are then compared over the range of CoCo-ization fractions $\beta \in [0,1]^2$ to determine the default scenarios with recovery rate $\alpha = \frac{1}{3}$.  Figure~\ref{fig:counterexample} provides images of the varying default scenarios under these varying network parameters.  We wish to note that changing the recovery rate does not affect the general shape of the defaulting regions but only the sum total of defaulting institutions.  However, the trigger level causes significant impacts to the defaulting regions. 
Intuitively, the higher the trigger level $\tau$, the smaller the shock necessary to trigger CoCo conversion. As these conversions result in markdowns in asset value for the debt holder, higher trigger levels can increase contagion and, thus, defaults.  
This indicates that for CoCo bonds to be utilized for financial stability, they should be implemented with low trigger levels wherever possible so that conversion only occurs as a final effort to preclude an even more costly default.
\begin{figure}[h]
\centering
\begin{subfigure}[t]{0.3\textwidth}
\centering
\includegraphics[width=\textwidth]{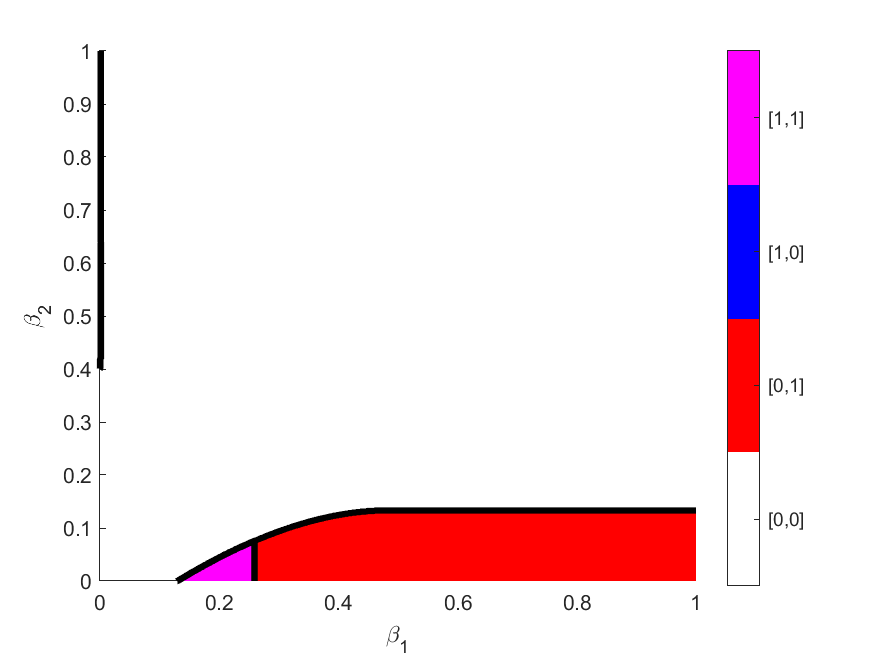}
\caption{$\tau = 1$}
\end{subfigure}
~
\begin{subfigure}[t]{0.3\textwidth}
\centering
\includegraphics[width=\textwidth]{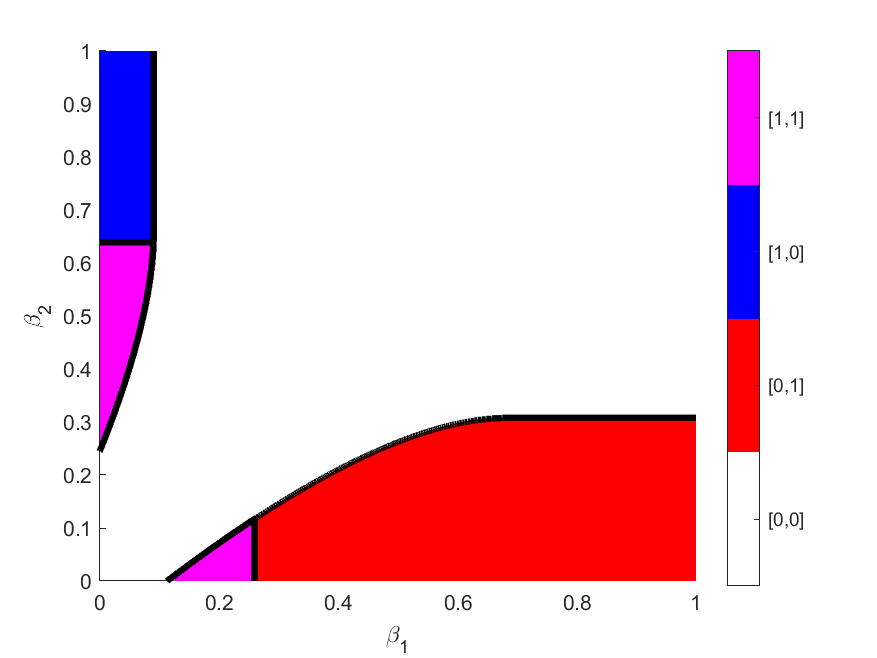}
\caption{$\tau = 2$}
\end{subfigure}
~
\begin{subfigure}[t]{0.3\textwidth}
\centering
\includegraphics[width=\textwidth]{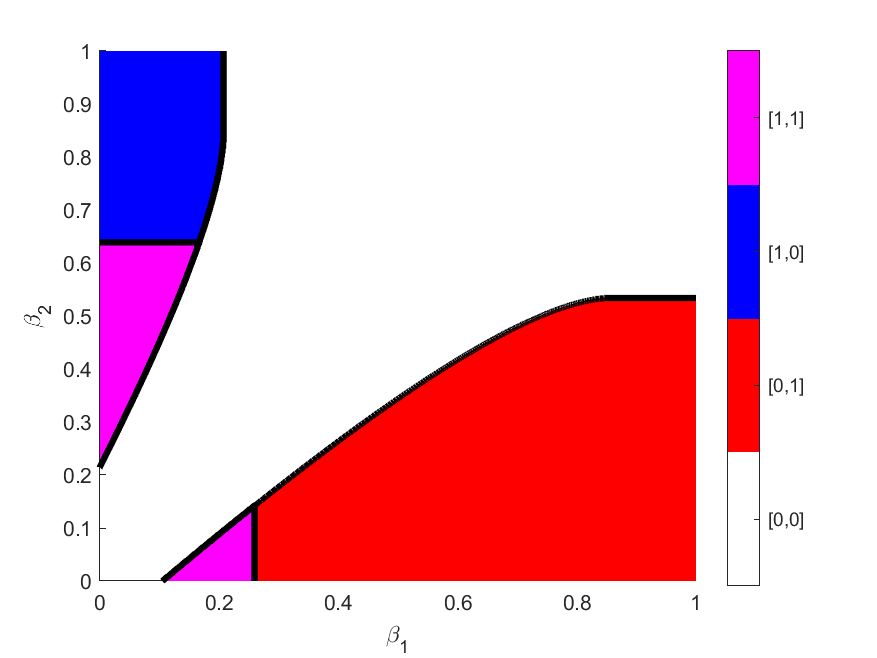}
\caption{$\tau = 5$}
\end{subfigure}
\caption{Example~\ref{ex:counterexample}: Impact of varying network parameters on the set of defaulting firms for a 2 bank system.}
\label{fig:counterexample}
\end{figure}
\end{example}

\section{European Banking Authority Case Studies}\label{sec:eba}
In this section we  study the implications of CoCo bonds on systemic risk in the EU system, using the data from the 2011 EBA stress test to calibrate the network of obligations\fn{See \url{https://eba.europa.eu/risk-analysis-and-data/eu-wide-stress-testing/2011/results}. Due to complications with the calibration methodology, we only consider 87 of the 90 institutions: DE029, LU45, and SI058 were not included in this analysis.}.  This data set, collected shortly before 
Europe's sovereign debt crisis, is of particular interest for financial contagion and has been studied in several papers under the default contagion model of~\cite{EN01}:  \cite{GV16,CLY14,feinstein2017currency,feinstein2017sensitivity}. Summary statistics of the system-wide balance sheet are provided in Table~\ref{table:EBA-data}. Figure~\ref{fig:EBA-data} shows a histogram of the size of these external obligations across the system of banks, as a fraction of the total external liabilities of  \euro23.381 \emph{trillion}. 
We note that the nature of interbank lending in Europe is such that interbank obligations are predominantly unsecured and thus the clearing payment system developed in this work is applicable.  

The EBA stress test data set identifies the total assets $TA_i$, capital $C_i$, and aggregated interbank liabilities $\sum_{j = 1}^n L_{ij}$ for all banks $i$ for a date in 2011. Since this is insufficient to construct the necessary stylized bank balance sheets for our model, we make additional assumptions. First, since equity cross-holdings between systemically important financial institutions are typically small, we assume the original equity cross-holding for every pair of firms $i,j$ has $\pi_{ij}^e = 0$.     Next, as  in \cite{CLY14,GY14}, we assume that aggregated interbank liabilities are equal to aggregated interbank assets,\fn{In actuality, as done in \cite{GV16} we perturb the interbank assets slightly so as to satisfy some technical conditions.} that all assets that are not interbank must be external, and that all liabilities that are not capital or interbank must be external. Thus, the external assets and liabilities are given by
\begin{align*}
x_i &:= TA_i - \sum_{j = 1}^n L_{ji} \quad \text{and} \quad L_{i0} := TA_i - \sum_{j = 1}^n L_{ij} - C_i\ .
\end{align*}
 Finally, we utilize the method of \cite{GV16} with parameters $p = 0.5$, $\text{thinning} = 10^4$, $n_{\text{burn-in}} = 10^9$, and $\lambda = \frac{p n (n-1)}{\sum_{i = 1}^n \sum_{j = 1}^n L_{ij}} \approx 0.00122$ to find a single realization of the liability matrix $L_{ij}$ consistent with the calibrated row and column sums.  

The following three case studies focus on the implications on systemic risk in the EU network of CoCo-izing the debts $L$ at varying trigger levels $\tau$, external and interbank CoCo fractions $\beta_0,\beta$, external shocks parametrized by $\xi$, and counterfactual levels of interbank debt $\gamma$. These studies all have fixed recovery rate $\alpha = \frac{1}{2}$ and conversion rate $q = 1$. All variable parameters are chosen to be constant across all banks due to the insights gained from Example~\ref{ex:counterexample} in which we determined that the heterogeneity of CoCo debts triggered adverse impacts. Broadening the scope of Section~\ref{sec:discussion}, we now look at three different systemic risk measures realized at maturity $T$:  the fraction of total external liabilities that are paid as debt and equity; the total value of the original shares; the total number of defaulted banks. Recall from Assumption~\ref{ass:maximal} that all results presented herein are w.r.t.\ the maximal clearing assets.

\begin{table}
\centering
\begin{tabular}{r|c|c|}
\cline{2-3}
& {\bf Total Assets (\euro\emph{trillion})} & {\bf Total Liabilities (\euro\emph{trillion})} \\ \cline{2-3}\noalign{\vskip\doublerulesep\vskip-\arrayrulewidth}\cline{2-3}
{\bf External:}~ & 24.383 & 23.381\\ \cline{2-3}
{\bf Interbank:}~ & 3.072 & 3.072\\ \cline{2-3}
{\bf Capital:}~ & -- & 1.002\\ \cline{2-3}
\end{tabular}
\caption{Summary statistics of the banking system from the 2011 EBA stress test data set.}
\label{table:EBA-data}
\end{table}

\begin{figure}[t]
\centering
\includegraphics[width=.6\textwidth]{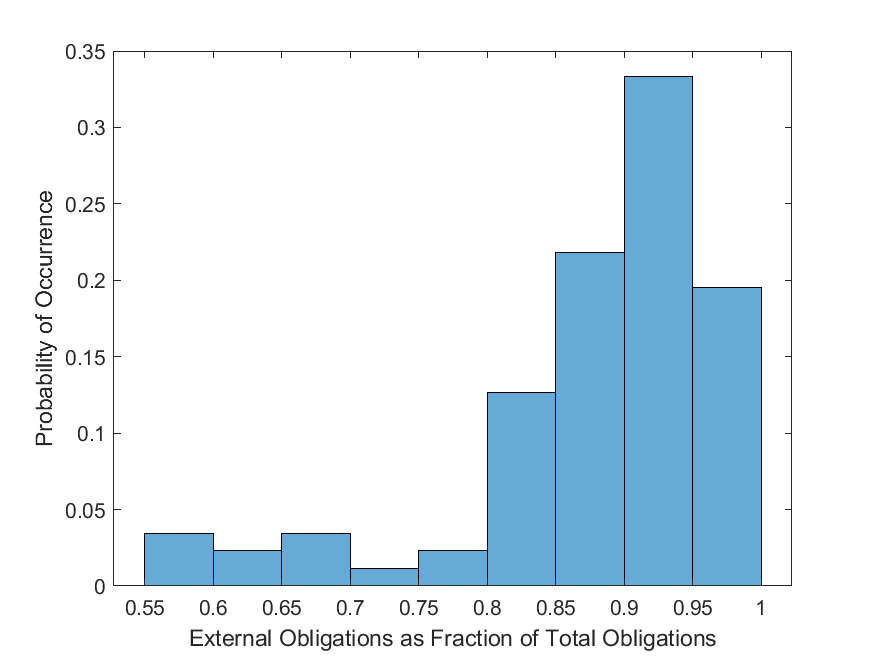}
\caption{Histogram of the fractional amounts that external liabilities make up of the total obligations, for all banks in the 2011 EBA stress test data set.}
\label{fig:EBA-data}
\end{figure}

\begin{study}\label{ex:EBA}
This study considers the calibrated EBA network model under a stress-test scenario where all external bank assets are subjected to a fractional decrease of $\xi=3\%$.   We consider two different CoCo-ization schemes with variable levels $\beta_0,\beta$ and trigger $\tau$: 
\begin{itemize}
\item \emph{Full CoCo-ization}: Equal fractions $\beta_0=\beta>0$ of all debts, both interbank and external, are CoCo-ized; and
\item \emph{Interbank CoCo-ization}: A fraction $\beta>0$ of interbank debts are CoCo-ized, and external debts remain vanilla ($\beta_0=0$).
\end{itemize}
We do not consider \emph{external CoCo-ization} with $\beta_0>0,\beta=0$ -- this scheme appears approximately equivalent to the full CoCo-ization scheme.  For the plain Eisenberg-Noe system with $\beta_0 =\beta = 0$ (i.e., when all debts are vanilla), the 3\% stress scenario  leads to the default of 72 of the 87 banks, with 51.48\% of debt owed to society repaid.  

Figures~\ref{fig:EBA-all} and~\ref{fig:EBA-bank} display the fractional repayment of external liabilities under full CoCo-ization and interbank CoCo-ization respectively.  
First we note the similarities between these two schemes to determine heuristics that may generally hold.  In particular, the behavior of the system worsens as the trigger level $\tau$ increases and as the level of CoCo-ization $\beta$ decreases.  
In this static setting, the best case scenario is for full repayment. For this reason the health of the system is monotonic in the trigger level $\tau$: the greater the trigger level $\tau$ the more often debts are converted and, as the equity has no time to grow, a write-down occurs.  
Now consider the full CoCo-ization scheme depicted in Figure~\ref{fig:EBA-all}.  For low levels of CoCo-ization, there are a significant number of defaults (as is the case in the purely Eisenberg-Noe setting with $\beta = 0$), but at a threshold level of $\beta \approx 1.2\%$ the cascade of defaults is eliminated and all banks become solvent with the external system (i.e., the real economy) recovering $99.84\%$ of the face value of its initial holdings.  However, even though no additional defaults occur, as more debts are CoCo-ized the fractional repayment to the external system decreases.
In contrast, for the interbank CoCo-ization scheme depicted in Figure~\ref{fig:EBA-bank}, the interaction between the level of CoCo-ization $\beta$ and the trigger level $\tau$ is more complicated.  For the computed trigger levels $\tau$, the external system recovers as much as $99.86\%$ repayment at $\beta \approx 11.7\%$ CoCo-ization. 
We conclude the discussion of payments by directly comparing the full CoCo-ization and interbank CoCo-ization schemes.  Clearly, for low CoCo-ization levels $\beta$, full CoCo-ization outperforms the interbank CoCo-ization by significant margins.  However, for high CoCo-ization levels $\beta$ the two schemes are comparable (in fact the interbank CoCo-ization slightly outperforms full CoCo-ization for these high levels of $\beta$).

In Figures~\ref{fig:EBA-all-equity} and~\ref{fig:EBA-bank-equity}, the value for the original shareholders is displayed under full CoCo-ization and interbank CoCo-ization respectively.
Notably, under full CoCo-ization, the original shareholders benefit under higher trigger levels $\tau$ \emph{and} greater CoCo-ization.  In fact, by comparing Figures~\ref{fig:EBA-all-equity} and~\ref{fig:EBA-all}, it can be observed that for CoCo-ization with high enough $\beta$ so that defaults are avoided, the original shareholders benefit at the expense of the bond holders and vice versa.
In contrast, interbank CoCo-ization can increase the original shareholder value only to a minimal degree, and does so in tandem with increasing benefits for the external debtholders as shown in Figures~\ref{fig:EBA-bank-equity} and~\ref{fig:EBA-bank}.
Clearly, the full CoCo-ization scheme benefits the original shareholders more than interbank CoCo-ization, but the original shareholders would prefer both schemes over the no CoCo-ization scheme ($\beta = 0$).
\begin{figure}[t]
\centering
\begin{subfigure}[t]{0.45\textwidth}
\centering
\includegraphics[width=\textwidth]{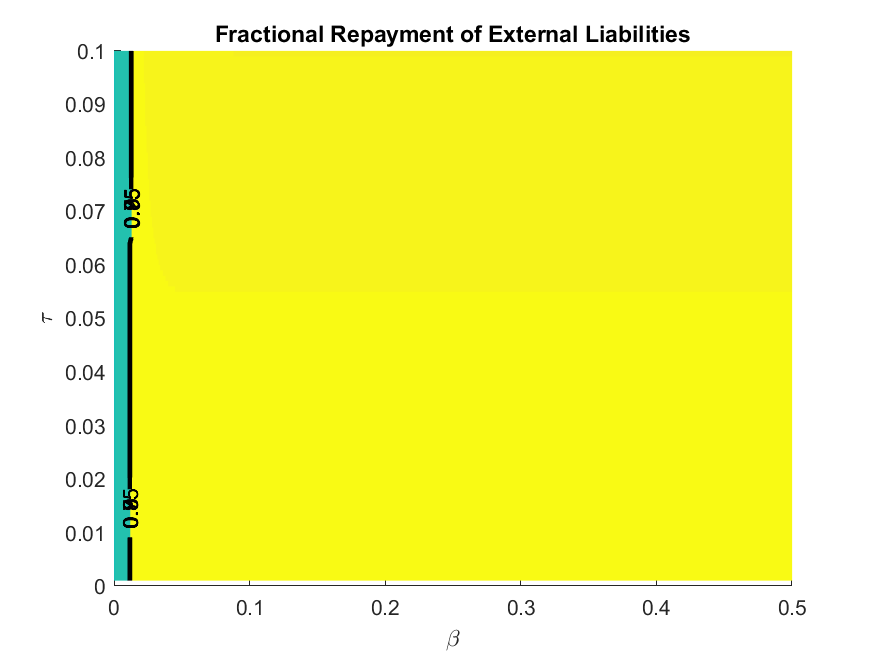}
\caption{Full CoCo-ization}
\label{fig:EBA-all}
\end{subfigure}
~
\begin{subfigure}[t]{0.45\textwidth}
\centering
\includegraphics[width=\textwidth]{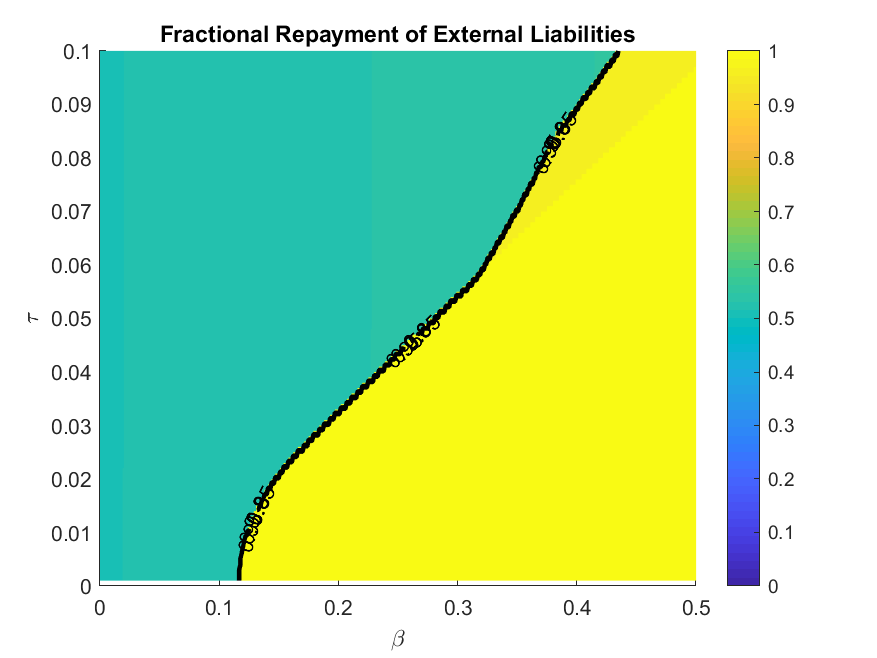}
\caption{Interbank CoCo-ization}
\label{fig:EBA-bank}
\end{subfigure}
\caption{The final fractional value of external liabilities for the EU network  versus the contingent convertible fraction $\beta$ and the trigger level $\tau$.}
\label{fig:EBA}
\end{figure}
\begin{figure}[t]
\centering
\begin{subfigure}[t]{0.45\textwidth}
\centering
\includegraphics[width=\textwidth]{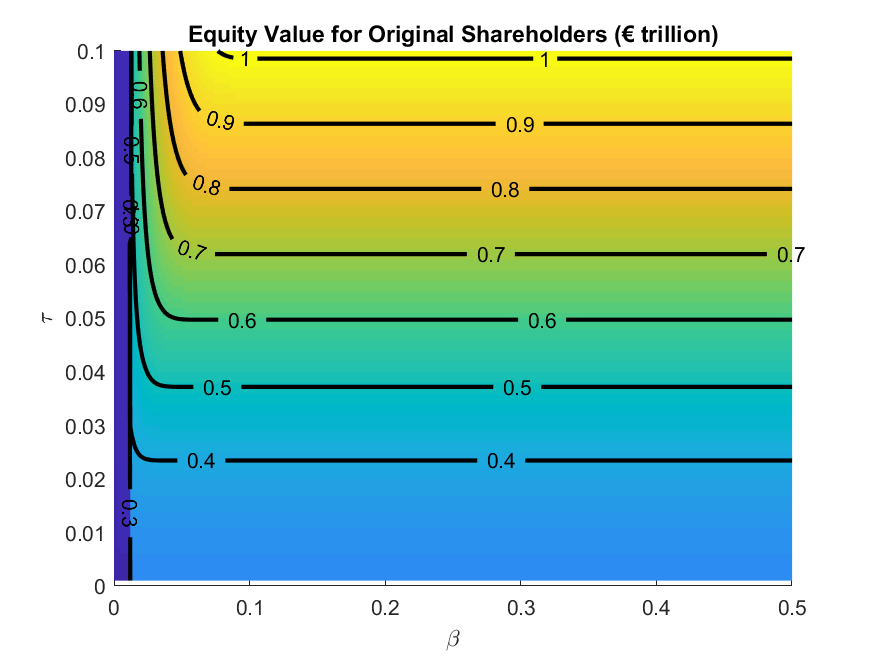}
\caption{Full CoCo-ization}
\label{fig:EBA-all-equity}
\end{subfigure}
~
\begin{subfigure}[t]{0.45\textwidth}
\centering
\includegraphics[width=\textwidth]{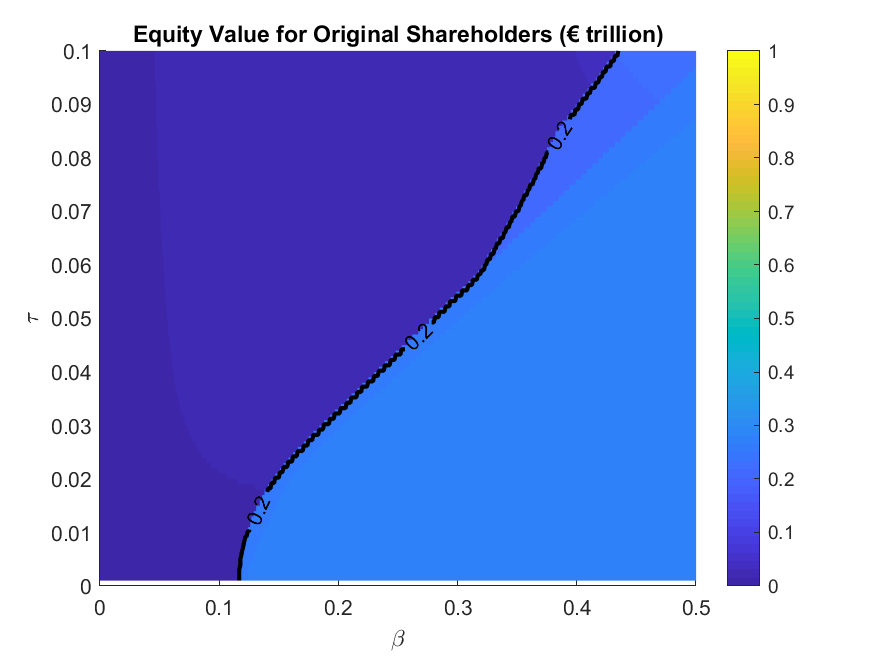}
\caption{Interbank CoCo-ization}
\label{fig:EBA-bank-equity}
\end{subfigure}
\caption{The final value of the original shares for the EU network  versus  the contingent convertible fraction $\beta$  and the trigger level $\tau$.}
\label{fig:EBA-equity}
\end{figure}

\end{study}

\begin{study}\label{ex:EBA-stress}
This study investigates the stability of the EU network when the external assets are subjected to a variable stress of fractional size $\xi$. It compares the no CoCo-ization scheme to three CoCo-ization schemes with trigger level at the Basel III tier 1 leverage requirement $\tau = 0.03$: 
\begin{itemize}
\item \emph{full CoCo-ization} $\beta_0=\beta=1$: all debts, both interbank and external, are CoCo-ized;
\item \emph{external CoCo-ization} $\beta_0=1,\beta=0$: all external debts are CoCo-ized and all interbank debts are vanilla;
\item \emph{interbank CoCo-ization} $\beta_0=0,\beta=1$: all interbank debts are CoCo-ized and all external debts are vanilla.\end{itemize}
The no CoCo-ization scheme with $\beta_0=\beta=0$ corresponds to the plain Eisenberg-Noe system.  For each fraction $\xi\in [0\%,10\%]$, a system-wide stress is applied leaving each bank with $(1-\xi)x$  in external assets. Figure~\ref{fig:EBA-stress} displays different systemic risk measures for the resultant equilibrium:  Figure~\ref{fig:EBA-stresspay} shows the fractional value of debts owed external to the system; Figure~\ref{fig:EBA-stressequity} plots the value for the original shareholders;  and Figure~\ref{fig:EBA-stressdefault} shows the number of defaulting banks. In general, we see that the system without CoCo-ization performs poorly, encountering a large number of defaults above the threshold $\xi =1.77\%$, with low debt repayment  and large losses for the original shareholders. The full CoCo-ization and external CoCo-ization schemes behave alike: All banks remain solvent at all stress levels and the losses to the external system are solely caused by the fractional conversion of CoCo debts (see Example~\ref{ex:setting-coco} for a simple construction for such losses).  Finally, the interbank CoCo-ization scheme performs almost as well as the other two schemes up to a sizeable stress level $\xi = 3.95\%$, albeit at the expense of a single bank defaulting.  Beyond $\xi = 3.95\%$,  interbank CoCo-ization exhibits significant default contagion, though it always outperforms the no CoCo-ization scheme in repayments of debts, equity for the original shareholders, and number of defaults.
\begin{figure}[t]
\centering
\begin{subfigure}[t]{0.3\textwidth}
\centering
\includegraphics[width=\textwidth]{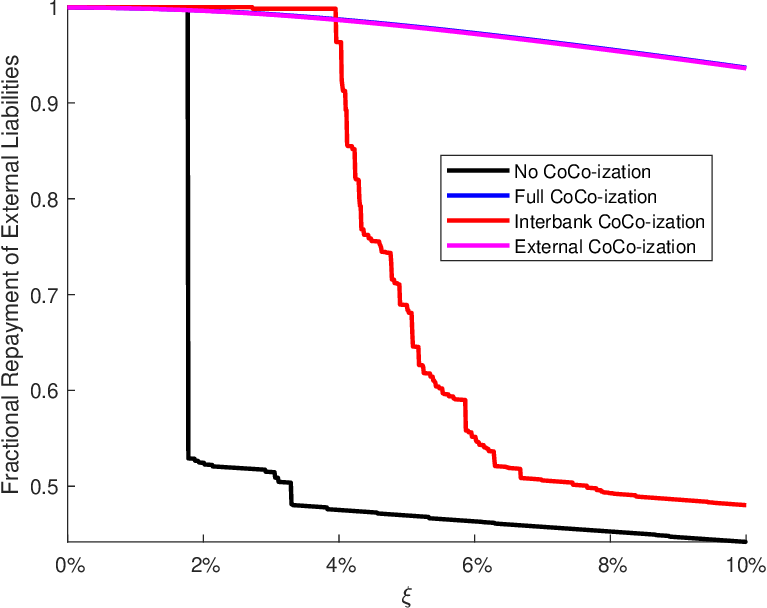}
\caption{Fractional value of the external debt}
\label{fig:EBA-stresspay}
\end{subfigure}
~
\begin{subfigure}[t]{0.3\textwidth}
\centering
\includegraphics[width=\textwidth]{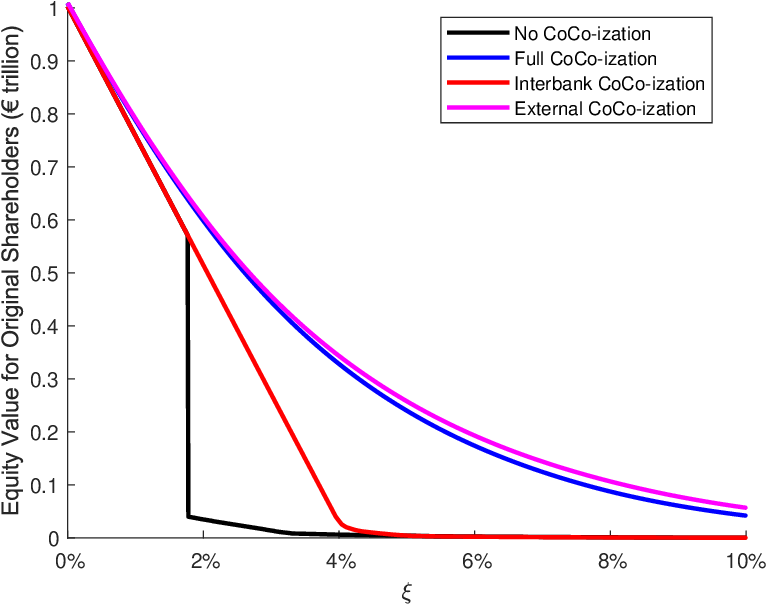}
\caption{Total value of the original shares}
\label{fig:EBA-stressequity}
\end{subfigure}
~
\begin{subfigure}[t]{0.3\textwidth}
\centering
\includegraphics[width=\textwidth]{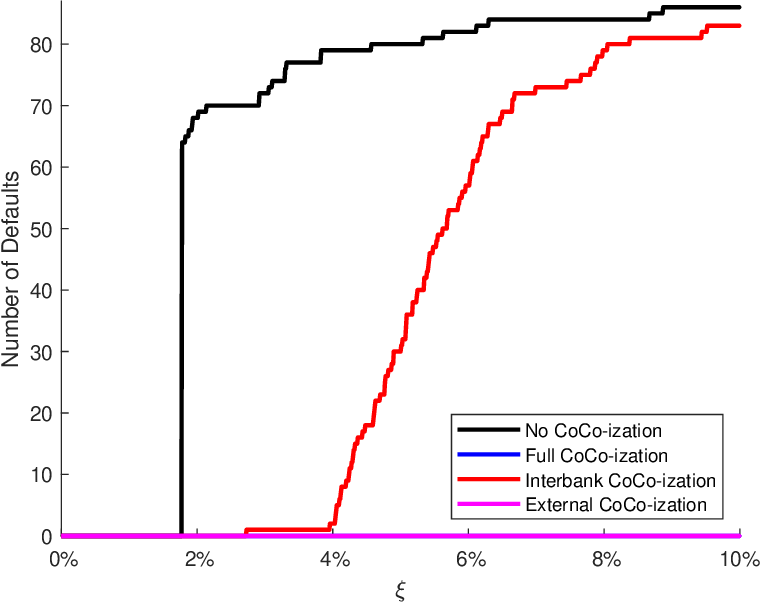}
\caption{Number of defaulted banks}
\label{fig:EBA-stressdefault}
\end{subfigure}
\caption{Three systemic risk measures are plotted for the EU network versus the fractional stress level $\xi$, under four different CoCo-ization schemes.}
\label{fig:EBA-stress}
\end{figure}
\end{study}

\begin{study}\label{ex:EBA-interbank}
This case study of the EU network investigates how systemic risk depends on the level of the interbank debt. We study the same three CoCo-ization schemes from Case Study~\ref{ex:EBA-stress} (again with $\tau = 0.03$ matching the Basel III tier 1 leverage requirement) but with a variable  fraction $\gamma \in (0,1]$ of the original external liabilities  replaced by interbank obligations.  That is,  for any $\gamma\le 1$, the total external liabilities is \euro$23.381(1-\gamma)$ \emph{trillion}, and the  total interbank assets and liabilities of \euro$(3.072+23.381\gamma)$ \emph{trillion} are distributed in the same proportions as in the original EU network.  The external assets are also reduced to keep the capital of each bank constant.  We then stress these (modified) external assets by a shock with $\xi=5\%$ (though not displayed, the results for stresses of $3\%$ as in Study~\ref{ex:EBA} or $10\%$, are similar). Our main motivation is to see whether interbank lending improves financial stability if the lending contract is a CoCo. 

Figure~\ref{fig:EBA-interbank} displays how three different systemic risk measures of the resultant equilibrium vary for $\gamma \in [0,1]$. As expected, the stressed system without CoCo-ization worsens as $\gamma$ increases,  exhibiting large numbers of defaults, low repayment under the stress scenario, and large losses for the original shareholders at all levels of $\gamma$.  In striking contrast, the interbank CoCo-ization scheme uniformly improves as the interbank obligations increase. Nearly all external liabilities paid in full (at over $99.83\%$, with only a single default) when $\gamma\ge 18.5\%$; all external liabilities are paid in full and no defaults occur when $\gamma \ge 35.8\%$. The original shareholders also uniformly benefit with increasing interbank CoCo-ization.  Full CoCo-ization and external CoCo-ization both exhibit increasing losses for the external debt holders as interbank obligations increase, and in the case of external CoCo-ization all banks eventually default as interbank obligations become large enough.  Intriguingly, the benefit to the original shareholders is non-monotonic: it increases with $\gamma$ only up to a threshold  of $\gamma = 63.0\%$ under full CoCo-ization and  $\gamma = 55.5\%$ under external CoCo-ization, and decreases thereafter.   Losses beyond  $\gamma = 55.5\%$ are extremely rapid for external CoCo-ization as the shareholders lose wealth from the default of banks.
\begin{figure}[t]
\centering
\begin{subfigure}[t]{0.3\textwidth}
\centering
\includegraphics[width=\textwidth]{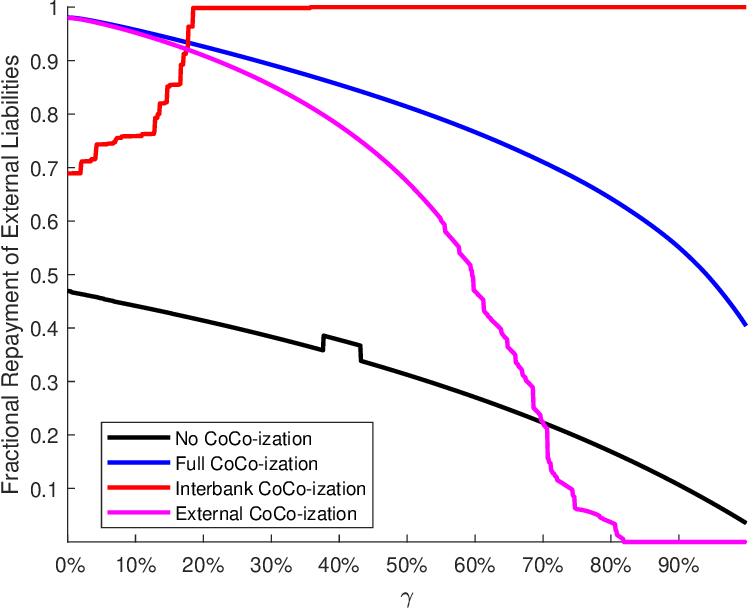}
\caption{Fractional value of the external debt}
\label{fig:EBA-fracpay}
\end{subfigure}
~
\begin{subfigure}[t]{0.3\textwidth}
\centering
\includegraphics[width=\textwidth]{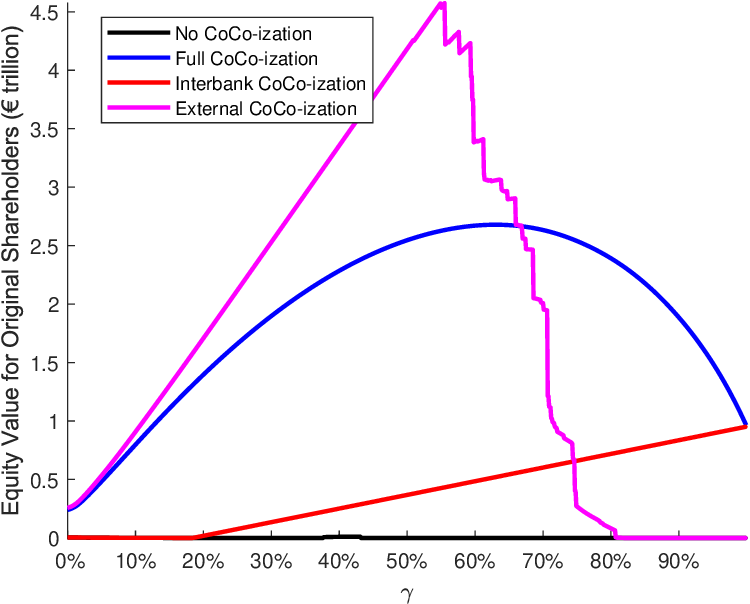}
\caption{Total value of the original shares}
\label{fig:EBA-fracequity}
\end{subfigure}
~
\begin{subfigure}[t]{0.3\textwidth}
\centering
\includegraphics[width=\textwidth]{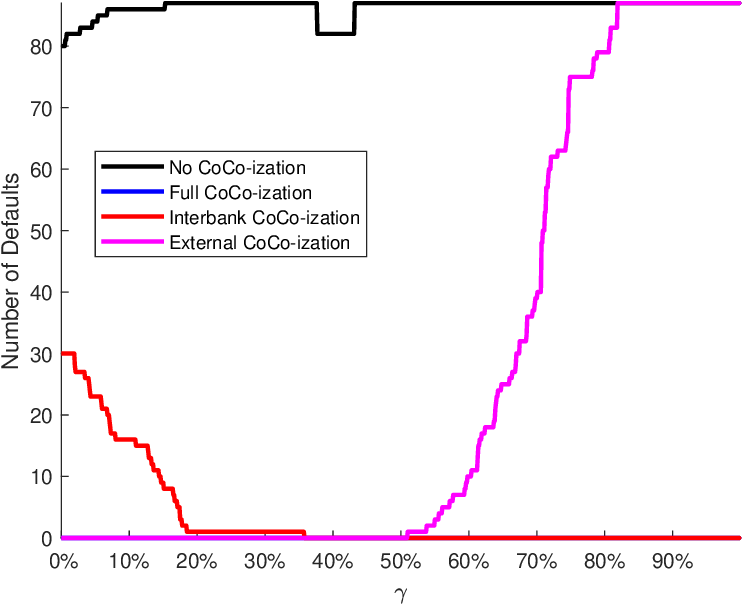}
\caption{Number of defaulted banks}
\label{fig:EBA-fracdefault}
\end{subfigure}
\caption{Three systemic risk measures are plotted for the EU network versus  the interbank fraction $\gamma$, under four different CoCo-ization schemes.}
\label{fig:EBA-interbank}
\end{figure}
\end{study}

\section{Conclusion}\label{sec:conclusion}

There can be no simple answer to the question ``Do CoCos reduce financial systemic risk?''.   The single firm market valuation studies of \cite{glasserman2012coco,AKB2013coco,SMS2018coco}, as well as our clearing analysis of Section~\ref{sec:1bank}, show that the answer is yes for fractional CoCo financing of a single bank. However, the counterexample given in Section~\ref{sec:discussion-counterexample} illustrates that no all-purpose theorem to this effect is provable in the very general CoCo network setting we introduce in Section~\ref{sec:nbank}. Nonetheless, the main contributions of Sections~\ref{sec:discussion} and~\ref{sec:eba} provide evidence that in typical networks, such as the EU system at the time of the 2011 EBA stress test, encouraging the use of standardized interbank CoCos in place of unsecured interbank lending could improve resilience to financial contagion. For example, our case studies in Section~\ref{sec:eba} have shown that replacing 100\% of the unsecured interbank lending in the EU system by standardized fractional CoCo contracts with trigger level $\tau=0.03$ and full conversion $q = 1$ is a win-win-win situation for financial stability: three different systemic risk measures show significant improvement over the vanilla funded system. The same case studies also provide evidence that externally held CoCos may not provide the same uniform improvements to financial stability that interbank CoCos provide. These observations suggest that future banking regulation needs to explore positive incentives to the use of standardized CoCo interbank lending contracts, with fractional conversion and a trigger level $\tau$ set near the Basel III  3\% tier 1 leverage ratio requirement. 

Justifying regulatory changes of this type would require much greater theoretical understanding of CoCos and systemic risk. First, researchers need to do much more extensive data-driven case studies of CoCo financing for other financial jurisdictions and other time periods. Second, to understand the pricing of CoCos and related securities in a financial system,    single firm market valuation should be extended to  CoCo generalizations of the dynamic vanilla network models of \cite{gourieroux2012}.  Such studies may reveal, for example, that the higher interest rates that need to be paid on CoCos with low triggers will limit their adoption by financial institutions. Thirdly, much more investigation is needed to determine parameters for a standardized CoCo that will effectively address financial systemic stability. Finally, one needs to explore theoretically how expanding the use of CoCos issued by financial institutions will impact global bond and repo markets.

\bibliographystyle{plainnat}
\bibliography{bibtex2}

\begin{thebibliography}{26}
\providecommand{\natexlab}[1]{#1}
\providecommand{\url}[1]{\texttt{#1}}
\expandafter\ifx\csname urlstyle\endcsname\relax
  \providecommand{\doi}[1]{doi: #1}\else
  \providecommand{\doi}{doi: \begingroup \urlstyle{rm}\Url}\fi

\bibitem[Acemoglu et~al.(2015)Acemoglu, Ozdaglar, and Tahbaz-Salehi]{AOT15}
Daron Acemoglu, Asuman Ozdaglar, and Alireza Tahbaz-Salehi.
\newblock Systemic risk and stability in financial networks.
\newblock \emph{American Economic Review}, 105\penalty0 (2):\penalty0 564--608,
  2015.

\bibitem[Avdjiev et~al.(2013)Avdjiev, Kartasheva, and Bogdanova]{AKB2013coco}
Stefan Avdjiev, Anastasia Kartasheva, and Bilyana Bogdanova.
\newblock {CoCos}: A primer.
\newblock \emph{BIS Quarterly Review}, page~43, September 2013.

\bibitem[Balter et~al.(2022)Balter, Schweizer, and Vera]{balter2020contingent}
Anne~G Balter, Nikolaus Schweizer, and Juan~C Vera.
\newblock Contingent capital with stock price triggers in interbank networks.
\newblock \emph{Mathematics of Operations Research}, 2022.
\newblock DOI: 10.1287/moor.2022.1278.

\bibitem[Banerjee and Feinstein(2019)]{banerjee2017insurance}
Tathagata Banerjee and Zachary Feinstein.
\newblock Impact of contingent payments on systemic risk in financial networks.
\newblock \emph{Mathematics and Financial Economics}, 13\penalty0 (4):\penalty0
  617--636, 2019.

\bibitem[Banerjee and Feinstein(2022)]{BF18comonotonic}
Tathagata Banerjee and Zachary Feinstein.
\newblock Pricing of debt and equity in a financial network with comonotonic
  endowments.
\newblock \emph{Operations Research}, 2022.
\newblock DOI: 10.1287/opre.2022.2275.

\bibitem[Banerjee et~al.(2022)Banerjee, Bernstein, and Feinstein]{BBF18}
Tathagata Banerjee, Alex Bernstein, and Zachary Feinstein.
\newblock Dynamic clearing and contagion in financial networks.
\newblock 2022.
\newblock Working paper.

\bibitem[Calice et~al.(2020)Calice, Sala, and Tantari]{calice2020contingent}
Giovanni Calice, Carlo Sala, and Daniele Tantari.
\newblock Contingent convertible bonds in financial networks.
\newblock 2020.
\newblock Working paper.

\bibitem[Chen et~al.(2016)Chen, Liu, and Yao]{CLY14}
Nan Chen, Xin Liu, and David~D. Yao.
\newblock An optimization view of financial systemic risk modeling: The network
  effect and the market liquidity effect.
\newblock \emph{Operations Research}, 64\penalty0 (5), 2016.

\bibitem[Eisenberg and Noe(2001)]{EN01}
Larry Eisenberg and Thomas~H. Noe.
\newblock Systemic risk in financial systems.
\newblock \emph{Management Science}, 47\penalty0 (2):\penalty0 236--249, 2001.

\bibitem[Feinstein(2019)]{feinstein2017currency}
Zachary Feinstein.
\newblock Obligations with physical delivery in a multi-layered financial
  network.
\newblock \emph{SIAM Journal on Financial Mathematics}, 10\penalty0
  (4):\penalty0 877--906, 2019.

\bibitem[Feinstein et~al.(2018)Feinstein, Pang, Rudloff, Schaanning, Sturm, and
  Wildman]{feinstein2017sensitivity}
Zachary Feinstein, Weijie Pang, Birgit Rudloff, Eric Schaanning, Stephan Sturm,
  and Mackenzie Wildman.
\newblock Sensitivity of the {E}isenberg and {N}oe clearing vector to
  individual interbank liabilities.
\newblock \emph{SIAM Journal on Financial Mathematics}, 9\penalty0
  (4):\penalty0 1286--1325, 2018.

\bibitem[Gandy and Veraart(2016)]{GV16}
Axel Gandy and Luitgard~A.M. Veraart.
\newblock A {B}ayesian methodology for systemic risk assessment in financial
  networks.
\newblock \emph{Management Science}, 63\penalty0 (12):\penalty0 4428--4446,
  2016.

\bibitem[Glasserman and Nouri(2012)]{glasserman2012coco}
Paul Glasserman and Behzad Nouri.
\newblock Contingent capital with a capital-ratio trigger.
\newblock \emph{Management Science}, 58\penalty0 (10):\penalty0 1816--1833,
  2012.

\bibitem[Glasserman and Young(2015)]{GY14}
Paul Glasserman and H.~Peyton Young.
\newblock How likely is contagion in financial networks?
\newblock \emph{Journal of Banking and Finance}, 50:\penalty0 383--399, 2015.

\bibitem[Gouri{\'e}roux et~al.(2012)Gouri{\'e}roux, H{\'e}am, and
  Monfort]{gourieroux2012}
Christian Gouri{\'e}roux, Jean-Cyprian H{\'e}am, and Alain Monfort.
\newblock Bilateral exposures and systemic solvency risk.
\newblock \emph{Canadian Journal of Economics}, 45\penalty0 (4):\penalty0
  1273--1309, 2012.

\bibitem[Gupta et~al.(2021)Gupta, Lu, and Wang]{gupta2019coco}
Aparna Gupta, Yueliang Lu, and Runzu Wang.
\newblock Addressing systemic risk using contingent convertible debt — a
  network analysis.
\newblock \emph{European Journal of Operational Research}, 290\penalty0
  (1):\penalty0 263--277, 2021.

\bibitem[Klages-Mundt and Minca(2020)]{minca2018reinsurance}
Ariah Klages-Mundt and Andreea Minca.
\newblock Cascading losses in reinsurance networks.
\newblock \emph{Management Science}, 66\penalty0 (9):\penalty0 4246--4268,
  2020.

\bibitem[Kusnetsov and Veraart(2019)]{KV16}
Michael Kusnetsov and Luitgard~A.M. Veraart.
\newblock Interbank clearing in financial networks with multiple maturities.
\newblock \emph{SIAM Journal on Financial Mathematics}, 10\penalty0
  (1):\penalty0 37--67, 2019.

\bibitem[Li et~al.(2022{\natexlab{a}})Li, Guo, and Meng]{li2020default}
Ping Li, Yanhong Guo, and Hui Meng.
\newblock The default contagion of contingent convertible bonds in financial
  network.
\newblock \emph{The North American Journal of Economics and Finance},
  60:\penalty0 101661, 2022{\natexlab{a}}.

\bibitem[Li et~al.(2022{\natexlab{b}})Li, Guo, and Meng]{li2020impact}
Ping Li, Yanhong Guo, and Hui Meng.
\newblock The impact of coco bonds on systemic risk considering liquidity risk.
\newblock \emph{Quantitative Finance}, 22\penalty0 (2):\penalty0 385--406,
  2022{\natexlab{b}}.

\bibitem[Rogers and Veraart(2013)]{RV13}
Leonard~C.G. Rogers and Luitgard~A.M. Veraart.
\newblock Failure and rescue in an interbank network.
\newblock \emph{Management Science}, 59\penalty0 (4):\penalty0 882--898, 2013.

\bibitem[Schuldenzucker et~al.(2017)Schuldenzucker, Seuken, and
  Battiston]{SSB16b}
Steffen Schuldenzucker, Sven Seuken, and Stefano Battiston.
\newblock {Finding Clearing Payments in Financial Networks with Credit Default
  Swaps is PPAD-complete}.
\newblock In Christos~H. Papadimitriou, editor, \emph{8th Innovations in
  Theoretical Computer Science Conference (ITCS 2017)}, volume~67 of
  \emph{Leibniz International Proceedings in Informatics (LIPIcs)}, pages
  32:1--32:20, Dagstuhl, Germany, 2017. Schloss Dagstuhl--Leibniz-Zentrum fuer
  Informatik.
\newblock ISBN 978-3-95977-029-3.
\newblock \doi{10.4230/LIPIcs.ITCS.2017.32}.

\bibitem[Schuldenzucker et~al.(2019)Schuldenzucker, Seuken, and
  Battiston]{SSB17}
Steffen Schuldenzucker, Sven Seuken, and Stefano Battiston.
\newblock Default ambiguity: Credit default swaps create new systemic risks in
  financial networks.
\newblock \emph{Management Science}, 66\penalty0 (5):\penalty0 1981--1998,
  2019.

\bibitem[Spiegeleer et~al.(2018)Spiegeleer, Marquet, and
  Schoutens]{SMS2018coco}
Jan~De Spiegeleer, Ine Marquet, and Wim Schoutens.
\newblock \emph{The Risk Management of Contingent Convertible (CoCo) Bonds}.
\newblock SpringerBriefs in Finance. Springer International Publishing, 2018.
\newblock ISBN 9783030018245.

\bibitem[Suzuki(2002)]{suzuki2002}
Teruyoshi Suzuki.
\newblock Valuing corporate debt: the effect of cross-holdings of stock and
  debt.
\newblock \emph{Journal of Operations Research}, 45\penalty0 (2):\penalty0
  123--144, 2002.

\bibitem[Weber and Weske(2017)]{AW_15}
Stefan Weber and Kerstin Weske.
\newblock The joint impact of bankruptcy costs, fire sales and cross-holdings
  on systemic risk in financial networks.
\newblock \emph{Probability, Uncertainty and Quantitative Risk}, 2\penalty0
  (1):\penalty0 9, 2017.

\end{thebibliography}

\pagebreak
\appendix
\begin{center}
{\LARGE Supplemental Material}
\end{center}

\section{Proof of Theorem~\ref{thm:cc-coco}}
Already the result has been proven for $\epsilon_2$ with $\alpha > 0$.  If $\alpha = 0$ then either $\epsilon_2 = \epsilon_1$ (as the defaults of the $d$ stressed banks are sufficient to force the unstressed banks into default as well) or $\epsilon_2$ \emph{does not exist} (as the immediate default of the $d$ stressed banks are insufficient to force the unstressed banks into default).  As such, the $\alpha = 0$ reduces to studying $\epsilon_1$.

Now we follow similar steps to investigate $\epsilon_1$, for which two equations hold:
\begin{align}
\label{eq:eps1-1} v_s &= x-\epsilon_1 +(d-1) A(v_s)+(n-d) A(v_n)-\bar p(\lambda(v_s))=0 \\
\label{eq:eps1-2} v_n &= x +d A(v_s)+(n-d-1)A(v_n)-\bar p(\lambda(v_n))>0
\end{align}
with $\lambda(v_s)=1$ and $A(v_s) = \frac{(1-\beta)z}{n-1}$. 
(Note that these equations do not depend on $\alpha$.)
To approach this problem, we reformulate \eqref{eq:eps1-1} and \eqref{eq:eps1-2} as
\begin{align}
\nonumber \epsilon_1 &= x + \frac{n-d}{n-1}p - (1-\beta_0)y-\frac{n-d}{n-1}(1-\beta)z\\
\label{eq:eps1-p} p &= (1-\beta\lambda(v_n))z + \pi^e(\lambda(v_n))\left[x + \frac{d}{n-1}(1-\beta)z + \frac{n-d-1}{n-1}p - \bar p(\lambda(v_n))\right]
\end{align}
which leads to $\nabla \epsilon_1 := \left(\frac{\partial \epsilon_1}{\partial \beta} \, , \, \frac{\partial \epsilon_1}{\partial \beta_0}\right)^\T = \left(\frac{n-d}{n-1}\left[\frac{\partial p}{\partial \beta} + z\right] \, , \, y + \frac{n-d}{n-1}\frac{\partial p}{\partial \beta_0}\right)^\T$.
There are three subcases to analyze: the $I_1$ case where $v_n \geq \tau (y+z)$ and $\lambda(v_n)=0$, the $I_2$ case where $v_n \in \tau ((1-\beta_0)y+(1-\beta)z \, , \, y+z)$ and $\lambda\in(0,1)$, and the $I_3$ case where $v_n \leq \tau ((1-\beta_0)y+(1-\beta)z)$ and $\lambda=1$.  
\begin{enumerate}
\item The $I_1$ case arises if $x\geq (1-\frac{(n-d-1)\pi^e}{n-1})\tau (y+z)+(y + \frac{d}{n-1}\beta z)$.  Since $\lambda(v_n)=0$, \eqref{eq:eps1-p} reduces to
\begin{align*}
p &= z + \pi^e [x + \frac{d}{n-1}(1-\beta)z + \frac{n-d-1}{n-1}p - y-z]\\
    &= \left(1 - \frac{n-d-1}{n-1}\pi^e\right)^{-1} \left(z + \pi^e[x +\frac{d}{n-1}(1-\beta)z -y-z]\right).
\end{align*}
As with $\epsilon_2$, taking the gradient of the second equation with respect to $\beta,\beta_0$ leads to
\begin{align*}
\frac{\partial p}{\partial \beta} &= -(1 - \frac{n-d-1}{n-1}\pi^e)^{-1}\pi^e\frac{d}{n-1}z \leq 0, \\
\frac{\partial v_n}{\partial \beta_0} &= 0.
\end{align*} 
Therefore the gradient $\nabla\epsilon_1$ is
\begin{align*}
\nabla\epsilon_1 &= \left(\frac{n-d}{n-1}\left(\left[\frac{1-\pi^e}{1-\frac{n-d-1}{n-1}\pi^e}\right]z\right) \, , \, y\right)^\T
\end{align*}
which shows $\nabla\epsilon_1 \in \bbr^2_+$.

\item The $I_3$ case arises if $x \leq (1 - \frac{(n-d-1)\pi^e(1)}{n-1})\tau[(1-\beta_0)y+(1-\beta)z] + (1-\beta_0)y$ for $\pi^e(1) = \frac{\beta z}{\beta_0 y + \beta z}c(1) + \pi^e (1-c(1))$ with $c(1) = q [1 - (\frac{(1-\beta_0)y + (1-\beta)z}{y+z})^{1/\tau}]$.  Since $\lambda(v_n)=1$, \eqref{eq:eps1-p} reduces to 
\begin{align*}
p &= (1-\beta)z + \pi^e(1)\left[x + \frac{d}{n-1}(1-\beta)z + \frac{n-d-1}{n-1}p - (1-\beta_0)y-(1-\beta)z\right]\\
    &= \left(1 - \frac{n-d-1}{n-1}\pi^e(1)\right)^{-1} \left((1-\beta)z + \pi^e(1)[x -(1-\beta_0)y - \frac{n-d-1}{n-1}(1-\beta)z]\right).
\end{align*}
As we will wish to determine the gradient of $\epsilon_1$ with respect to $\beta,\beta_0$ as above, let us first consider the gradient of $\pi^e(1)$
\begin{align*}
\frac{\partial \pi^e(1)}{\partial \beta} &= \frac{\beta_0 yz}{(\beta_0 y + \beta z)^2}c(1) + \frac{q}{\tau}\left(\frac{\beta z}{\beta_0 y + \beta z} - \pi^e\right)\frac{z}{(1-\beta_0)y+(1-\beta)z} (1 - c(1)),\\
\frac{\partial \pi^e(1)}{\partial \beta_0} &= -\frac{\beta yz}{(\beta_0 y + \beta z)^2}c(1) + \frac{q}{\tau}\left(\frac{\beta z}{\beta_0 y + \beta z} - \pi^e\right)\frac{y}{(1-\beta_0)y+(1-\beta)z} (1 - c(1)).
\end{align*}
We will first consider the derivative with respect to $\beta$ followed by $\beta_0$.
    \begin{itemize}
    \item First, consider the derivative of $p$ with respect to $\beta$
        \begin{align*}
        \frac{\partial p}{\partial \beta} &= \frac{-z + \frac{n-d-1}{n-1}\pi^e(1)z + \frac{\partial \pi^e(1)}{\partial \beta}\left(x - (1-\beta_0)y - \frac{n-d-1}{n-1}(1-\beta)z\right)}{1 - \frac{n-d-1}{n-1}\pi^e(1)}\\
            &\qquad + \frac{\left((1-\beta)z + \pi^e(1)\left[x - (1-\beta_0)y - \frac{n-d-1}{n-1}(1-\beta)z\right]\right)\frac{n-d-1}{n-1}\frac{\partial \pi^e(1)}{\partial \beta}}{\left(1 - \frac{n-d-1}{n-1}\pi^e(1)\right)^2}\\
            &= -z + \frac{x - (1-\beta_0)y}{\left(1 - \frac{n-d-1}{n-1}\pi^e(1)\right)^2}\frac{\partial \pi^e(1)}{\partial \beta}.
        \end{align*}
        Therefore $\frac{\partial \epsilon_1}{\partial \beta} \geq 0$ if and only if $\frac{\partial \pi^e(1)}{\partial \beta} \geq 0$. 
        Let us now consider $\frac{\partial \pi^e(1)}{\partial \beta}$; this is nonnegative if, and only if,
        \begin{align}
        \label{eq:Gamma} \pi^e \leq \frac{\beta z}{\beta_0 y + \beta z} + \frac{\tau}{q}\frac{(1-\beta_0)y+(1-\beta)z}{\beta_0 y + \beta z}\left[\left(\frac{y+z}{(1-\beta_0)y+(1-\beta)z}\right)^{\frac{q}{\tau}}-1\right]\left(1-\frac{\beta z}{\beta_0 y + \beta z}\right) =: \Gamma.
        \end{align}
        This is, trivially, true if $\pi^e \leq \frac{\beta z}{\beta_0 y + \beta z}$; assume that is not true (i.e., by assumption $\tau \leq q$).
        Let $\zeta := \frac{(1-\beta_0)y+(1-\beta)z}{\beta_0 y + \beta z}$, then
        \begin{align*}
        \Gamma &= \frac{\beta z}{\beta_0 y + \beta z} + \frac{\tau}{q}\zeta\left[\left(1 + \frac{1}{\zeta}\right)^{\frac{q}{\tau}}-1\right]\left(1-\frac{\beta z}{\beta_0 y + \beta z}\right).
        \end{align*}
        Therefore if $\xi := \frac{\tau}{q}\zeta\left[\left(1 + \frac{1}{\zeta}\right)^{\frac{q}{\tau}}-1\right] \geq 1$ for any $\tau \in (0,q]$ and $\zeta \geq 0$ then $\Gamma \geq 1$ and $\frac{\partial \pi^e(1)}{\partial \beta} \geq 0$.  Notably, for $\tau \in (0,q]$, $\xi$ is decreasing in $\zeta$.  As such, we need only consider the limit as $\zeta$ tends towards $\infty$; in fact, $\lim_{\zeta \to \infty} \xi = 1$.
    \item First, consider the derivative of $p$ with respect to $\beta_0$
        \begin{align*}
        \frac{\partial p}{\partial \beta_0} &= \frac{\pi^e(1)y + \frac{\partial \pi^e(1)}{\partial \beta_0}\left(x - (1-\beta_0)y - \frac{n-d-1}{n-1}(1-\beta)z\right)}{1 - \frac{n-d-1}{n-1}\pi^e(1)}\\
            &\qquad + \frac{\left((1-\beta)z + \pi^e(1)\left[x - (1-\beta_0)y - \frac{n-d-1}{n-1}(1-\beta)z\right]\right)\frac{n-d-1}{n-1}\frac{\partial \pi^e(1)}{\partial \beta_0}}{\left(1 - \frac{n-d-1}{n-1}\pi^e(1)\right)^2}\\
            &= \frac{\pi^e(1)y}{1 - \frac{n-d-1}{n-1}\pi^e(1)} + \frac{x - (1-\beta_0)y}{\left(1 - \frac{n-d-1}{n-1}\pi^e(1)\right)^2}\frac{\partial \pi^e(1)}{\partial \beta_0}.
        \end{align*}
        Therefore $\frac{\partial \epsilon_1}{\partial \beta_0} \geq 0$ if $\frac{\partial \pi^e(1)}{\partial \beta_0} \geq 0$.  Assume now that $\frac{\partial \pi^e(1)}{\partial \beta_0} < 0$.
        \begin{align*}
        \frac{\partial p}{\partial \beta_0} &= \frac{\pi^e(1)y}{1 - \frac{n-d-1}{n-1}\pi^e(1)} + \frac{x - (1-\beta_0)y}{\left(1 - \frac{n-d-1}{n-1}\pi^e(1)\right)^2}\frac{\partial \pi^e(1)}{\partial \beta_0}\\
            &\geq \frac{1}{1-\frac{n-d-1}{n-1}\pi^e(1)}\left(\pi^e(1)y + \tau\left[(1-\beta_0)y+(1-\beta)z\right]\frac{\partial \pi^e(1)}{\partial \beta_0}\right)\\
            &= \frac{y}{1-\frac{n-d-1}{n-1}\pi^e(1)}\left[\pi^e + \left(\frac{\beta z}{\beta_0 y + \beta z} - \pi^e\right)[q+(1-q)c(1)]\right.\\ 
            &\qquad \left.- \frac{\tau \left((1-\beta_0)y+(1-\beta)z\right)}{\beta_0 y + \beta z} \frac{\beta z}{\beta_0 y + \beta z} c(1)\right]\\
            &= \frac{y}{1-\frac{n-d-1}{n-1}\pi^e(1)}\bigg[(1-q)(1-c(1))\pi^e\\ 
            &\qquad \left.+ \left(q+(1-q)c(1) - \frac{\tau\left((1-\beta_0)y + (1-\beta)z\right)}{\beta_0 y + \beta z}c(1)\right)\frac{\beta z}{\beta_0 y + \beta z}\right]\\
            &= \frac{y}{1-\frac{n-d-1}{n-1}\pi^e(1)}\left[(1-q)(1-c(1))\pi^e + \left(1 - (1-q)\left(\frac{(1-\beta_0)y + (1-\beta)z}{y+z}\right)^{\frac{q}{\tau}}\right.\right.\\
            &\qquad \left.\left.-\tau\frac{(1-\beta_0)y+(1-\beta)z}{\beta_0 y + \beta z}\left(1 - \left(\frac{(1-\beta_0)y+(1-\beta)z}{y+z}\right)^{\frac{q}{\tau}}\right)\right)\frac{\beta z}{\beta_0 y + \beta z}\right]
        \end{align*}
        Therefore $\frac{\partial \epsilon_1}{\partial \beta_0} \geq 0$ if $\Psi := \tau\frac{(1-\beta_0)y+(1-\beta)z}{\beta_0 y + \beta z}\left[1 - \left(\frac{(1-\beta_0)y+(1-\beta)z}{y+z}\right)^{\frac{q}{\tau}}\right] \leq q$.
        First, notice that
        \begin{align*}
        \frac{\partial \Psi}{\partial \tau} &= \frac{(1-\beta_0)y+(1-\beta)z}{\beta_0 y + \beta z}g\left(\frac{(1-\beta_0)y+(1-\beta)z}{y+z}\right)\\
        g(\gamma) &= 1 - \gamma^{\frac{q}{\tau}} + \frac{q}{\tau}\gamma^{\frac{q}{\tau}}\log(\gamma).
        \end{align*}
        In fact, we only need consider $g$ with domain $[0,1]$ with $g(0) = 1$, $g(1) = 0$, and $g'(\gamma) = \frac{q}{\tau^2}\gamma^{\frac{q}{\tau}-1}\log(\gamma) \leq 0$, i.e., $\frac{\partial \Psi}{\partial \tau} \geq 0$.  As such, $\Psi \leq q$ if $\lim_{\tau \to \infty} \Psi \leq q$.
        \begin{align*}
        \lim_{\tau \to \infty} \Psi &= q h(\beta_0 y + \beta z), \quad h(\xi) := -\frac{y+z-\xi}{\xi}\log\left(\frac{y+z-\xi}{y+z}\right) \quad \forall \xi \in [0,y+z].
        \end{align*}
        It can, further, be shown that $\sup_{\xi \in [0,y+z]} h(\xi) = 1$ and the result is proven.
    \end{itemize}

\item Finally, consider the $I_2$ case.  Herein, we wish to consider $p$ explicitly in terms of $v_n$, i.e.,
    \begin{align*}
    p &= (1-\beta\lambda_n)z + \pi^e(\lambda_n)v_n\\
        &= \frac{(\beta_0-\beta)yz}{\beta_0 y + \beta z} + \frac{1 + \tau}{\tau} \frac{\beta z}{\beta_0 y + \beta z} v_n - \frac{1}{\left(\tau[y+z]\right)^{\frac{q}{\tau}}}\left(\frac{\beta z}{\beta_0 y + \beta z} - \pi^e\right)v_n^{\frac{q+\tau}{\tau}}
    \end{align*}
    since, given the $I_2$ setting,
    \begin{align*}
    \lambda_n &= \frac{y+z}{\beta_0 y + \beta z} - \frac{v_n}{\tau[\beta_0 y + \beta z]}\\
    \pi^e(\lambda_n) &= \pi^e + \left(\frac{\beta z}{\beta_0 y + \beta z} - \pi^e\right)\left[1 - \left(\frac{v_n}{\tau(y+z)}\right)^{\frac{q}{\tau}}\right].
    \end{align*}
    Additionally, following the construction of $A$, \eqref{eq:eps1-2} reduces to
    \begin{align*}
    v_n &= \frac{\tau}{1+\tau}\left(x + \frac{d}{n-1}(1-\beta)z + \frac{n-d-1}{n-1}p\right)
    \end{align*} 
    or, equivalently, $p = \frac{n-1}{n-d-1}\left[\frac{1+\tau}{\tau}v_n - \left(x + \frac{d}{n-1}(1-\beta)z\right)\right]$.
    Therefore, by construction, $\nabla\epsilon_1 = \left(\frac{n-d}{n-d-1}\left(\frac{1+\tau}{\tau}\frac{\partial v_n}{\partial \beta}+z\right) \, , \, y + \frac{n-d}{n-d-1}\frac{1+\tau}{\tau}\frac{\partial v_n}{\partial \beta_0}\right)^\T$.
    We will first consider the derivative with respect to $\beta$ followed by $\beta_0$.
    \begin{itemize}
    \item First, consider the derivative of $v_n$ with respect to $\beta$
        \begin{align*}
        &\underbrace{\left(1 - \frac{n-d-1}{n-1}\left[\frac{\beta z}{\beta_0 y + \beta z}\left(1 - \frac{q+\tau}{1+\tau}\left(\frac{v_n}{\tau[y+z]}\right)^{\frac{q}{\tau}}\right) + \pi^e\left(\frac{q+\tau}{1+\tau}\left(\frac{v_n}{\tau[y+z]}\right)^{\frac{q}{\tau}}\right)\right]\right)}_{G\left(\frac{v_n}{\tau[y+z]}\right)} \frac{\partial v_n}{\partial \beta}\\
        &= \frac{\tau}{1+\tau}\underbrace{\left(-\frac{d}{n-1}z - \frac{n-d-1}{n-1}\frac{\beta_0 yz}{\left(\beta_0 y + \beta z\right)^2}\left[(y+z) - \frac{1+\tau}{\tau}v_n + \frac{1}{\left(\tau[y+z]\right)^{\frac{q}{\tau}}}v_n^{\frac{q+\tau}{\tau}}\right]\right)}_{F_\beta(v_n)}
        \end{align*}
        Before studying $\frac{\partial \epsilon_1}{\partial \beta}$, we want to consider the mapping $G$.  In particular, for $G$, we want to consider a domain of $[0,1]$; we find that: $G(0) = 1 - \frac{n-d-1}{n-1}\frac{\beta z}{\beta_0 y + \beta z} \in (0,1]$; $G(1) = 1 - \frac{n-d-1}{n-1}\left(\frac{1-q}{1+\tau}\frac{\beta z}{\beta_0 y + \beta z} + \left(1 - \frac{1-q}{1+\tau}\right)\pi^e\right) \in (0,1]$; $G'(\gamma) = \frac{n-d-1}{n-1}\frac{q}{\tau}\frac{q+\tau}{1+\tau}\left(\frac{\beta z}{\beta_0 y + \beta z} - \pi^e\right)\gamma^{\frac{q}{\tau}-1}$ keeps a constant sign for any $\gamma \in [0,1]$.  Therefore $G(\gamma) \in (0,1]$ for any $\gamma \in [0,1]$.
        Therefore, $\frac{\partial \epsilon_1}{\partial \beta} \geq 0$ if $F_\beta(v) + G(\frac{v}{\tau[y+z]})z \geq 0$ for any $v \in [\tau((1-\beta_0)y+(1-\beta)z) , \tau(y+z)]$.
        \begin{align*}
        &F_\beta(v) + G(\frac{v}{\tau[y+z]})z \\
        &= z \left[-\frac{d}{n-1} - \frac{n-d-1}{n-1}\frac{\beta_0 y}{(\beta_0 y + \beta z)^2}\left[(y+z) - \frac{1+\tau}{\tau}v + \frac{1}{(\tau[y+z])^{\frac{q}{\tau}}}v^{\frac{q+\tau}{\tau}}\right]\right.\\
        &\qquad \left. + 1 - \frac{n-d-1}{n-1}\left[\frac{\beta z}{\beta_0 y + \beta z}\left(1 - \frac{q+\tau}{1+\tau}\left(\frac{v}{\tau(y+z)}\right)^{\frac{q}{\tau}}\right) + \pi^e\frac{q+\tau}{1+\tau}\left(\frac{v}{\tau(y+z)}\right)^{\frac{q}{\tau}}\right]\right]\\
        &= z \frac{n-d-1}{n-1}\left[1 - \left(\frac{\beta z}{\beta_0 y + \beta z} + \left(\frac{(y+z)-v/\tau}{\beta_0 y + \beta z}\right)\frac{\beta_0 y}{\beta_0 y + \beta z}\right)\right.\\
        &\qquad \left.+ \frac{\beta_0 y}{(\beta_0 y + \beta z)^2}v\left(1 - \left(\frac{v}{\tau(y+z)}\right)^{\frac{q}{\tau}}\right) + \frac{q+\tau}{1+\tau}\left(\frac{\beta z}{\beta_0 y + \beta z} - \pi^e\right)\left(\frac{v}{\tau(y+z)}\right)^{\frac{q}{\tau}}\right].
        \end{align*}
        This is, trivially, nonnegative if $\pi^e \leq \frac{\beta z}{\beta_0 y + \beta z}$; assume that is not true (i.e., by assumption $\tau \leq q$).  We will complete this proof in three parts. First, $F_\beta(v) + G(\frac{v}{\tau[y+z]})z \geq 0$ at $v = \tau((1-\beta_0)+(1-\beta)z)$ if, and only if, $\pi^e \leq \Gamma$ (as defined in~\eqref{eq:Gamma}); as determined above, this always holds.  Second, $F_\beta(v) + G(\frac{v}{\tau[y+z]})z \geq 0$ at $v = \tau(y+z)$; by substitution
        \begin{align*}
        F_\beta(\tau(y+z)) + G(1)z &= z \frac{n-d-1}{n-1}\left[1 - \left(\frac{1-q}{1+\tau}\frac{\beta z}{\beta_0 y + \beta z} + \left(1 - \frac{1-q}{1+\tau}\right)\pi^e\right)\right] \geq 0 \ .
        \end{align*}  
        Third, we wish to show that $F_\beta(v) + G(\frac{v}{\tau[y+z]})z$ is minimized at either $\tau((1-\beta_0)y+(1-\beta)z)$ or $\tau(y+z)$. 
        By the above formulation, $F_\beta(v) + G(\frac{v}{\tau[y+z]})z \geq 0$ if, and only if,
        \begin{align*}
        \Theta(v) &:= 1 - \left(\frac{\beta z}{\beta_0 y + \beta z} + \left(\frac{(y+z)-v/\tau}{\beta_0 y + \beta z}\right)\frac{\beta_0 y}{\beta_0 y + \beta z}\right)\\
        &\qquad + \frac{\beta_0 y}{(\beta_0 y + \beta z)^2}v\left(1 - \left(\frac{v}{\tau(y+z)}\right)^{\frac{q}{\tau}}\right) + \frac{q+\tau}{1+\tau}\left(\frac{\beta z}{\beta_0 y + \beta z} - \pi^e\right)\left(\frac{v}{\tau(y+z)}\right)^{\frac{q}{\tau}} \geq 0 \ .
        \end{align*}
        By taking derivatives, $\frac{\partial \Theta}{\partial v} \geq 0$ if, and only if,
        \begin{align*}
        \pi^e \leq \frac{\beta z}{\beta_0 y + \beta z} + \frac{1+\tau}{q}\underbrace{\left[\frac{1+\tau}{q+\tau}\left(\frac{\tau(y+z)}{v}\right)^{\frac{q}{\tau}} - 1\right]}_{\rho(v)}\frac{\beta_0 y}{(\beta_0 y + \beta z)^2}.
        \end{align*}
        Since $\frac{\partial \rho}{\partial v} < 0$, it follows that
        \begin{align*}
        &\inf_{v \in [\tau((1-\beta_0)y+(1-\beta)z),\tau(y+z)]} F_\beta(v) + G(\frac{v}{\tau[y+z]})z\\ 
        &= \min\{F_\beta(\tau((1-\beta_0)y+(1-\beta)z)) + G(\frac{(1-\beta_0)y+(1-\beta)z}{y+z})z \, , \, F_\beta(\tau(y+z)) + G(1)z\} \geq 0 
        \end{align*}
        and the result is proven.
    
    \item First, consider the derivative of $v_n$ with respect to $\beta_0$
        \begin{align*}
        &G\left(\frac{v_n}{\tau[y+z]}\right) \frac{\partial v_n}{\partial \beta} = \frac{\tau}{1+\tau}\underbrace{\frac{n-d-1}{n-1}\frac{\beta y z}{\left(\beta_0 y + \beta z\right)^2}\left[(y+z) - \frac{1+\tau}{\tau}v_n + \frac{1}{\left(\tau[y+z]\right)^{\frac{q}{\tau}}}v_n^{\frac{q+\tau}{\tau}}\right]}_{F_{\beta_0}(v_n)}
        \end{align*}
        where $G: [0,1] \to (0,1]$ is defined above.  Consider now $F_{\beta_0}$ on the domain $[0,\tau(y+z)]$
        \begin{align*}
        F_{\beta_0}(\tau(y+z)) &= 0 \\ 
        F_{\beta_0}'(v) &= \frac{n-d-1}{n-1}\frac{1}{\tau}\frac{\beta yz}{(\beta_0 y + \beta z)^2}\left[-(1+\tau) + (q+\tau)\left(\frac{v}{\tau(y+z)}\right)^{\frac{q}{\tau}}\right]\\
            &= -\frac{n-d-1}{n-1}\frac{1}{\tau}\frac{\beta yz}{(\beta_0 y + \beta z)^2}\left[(1-q) + (q+\tau)\left(1 - \left(\frac{v}{\tau(y+z)}\right)^{\frac{q}{\tau}}\right)\right] \leq 0 \ .
        \end{align*}
        Therefore, it immediately follows that $\frac{\partial v_n}{\partial \beta_0} \geq 0$ and, thus, also $\frac{\partial \epsilon_1}{\partial \beta_0} \geq 0$.

    \end{itemize}

\end{enumerate}

\end{document}